\documentclass[11pt,a4paper]{article}

\usepackage[margin=2.8cm]{geometry}
\usepackage{amsmath,amssymb,amsthm,mathtools}
\usepackage{enumitem}
\usepackage{microtype}
\usepackage[hidelinks]{hyperref}
\usepackage{xcolor}
\usepackage{tikz}
\usetikzlibrary{arrows.meta,calc,decorations.markings}

\newtheorem{theorem}{Theorem}[section]
\newtheorem{lemma}[theorem]{Lemma}

\newtheorem{proposition}[theorem]{Proposition}
\newtheorem{corollary}[theorem]{Corollary}
\theoremstyle{remark}
\newtheorem{remark}[theorem]{Remark}
\theoremstyle{plain}

\newcommand{\HS}{\mathrm{HS}}
\newcommand{\cH}{\mathcal H}
\newcommand{\cS}{\mathcal  S}
\newcommand{\cB}{\mathcal B}
\newcommand{\cA}{\mathcal A}
\newcommand{\cW}{\mathcal W}
\newcommand{\dd}{\mathrm d}
\newcommand{\tr}{\operatorname{tr}}

\newcommand{\cX}{\mathcal X}
\newcommand{\cF}{\mathcal F}
\newcommand{\cI}{\mathcal I}
\newcommand{\bE}{\mathbb E}
\newcommand{\Var}{\operatorname{Var}}
\newcommand{\supp}{\operatorname{supp}}

\title{Information Geometry of the Geodesic Quantum $f$-Divergences}
\author{%
 Ángela Capel$^{1,2}$, Pablo Costa Rico$^{3}$\\[0.8em]
 \parbox{0.92\textwidth}{\centering\small
 $^{1}$Department of Applied Mathematics and Theoretical Physics, University of Cambridge,  Cambridge CB3 0WA, United Kingdom \\
 $^{2}$Fachbereich Mathematik, Universit\"at T\"ubingen, 72076 T\"ubingen, Germany\\
 $^{3}$Institute for Quantum Information, RWTH Aachen University,
 Aachen, Germany\\
 }}
\date{}

\begin{document}
\maketitle
\vspace{-2.3em}

\begin{abstract}
We study the differential, statistical, and geometrical consequences generated by the geodesic quantum $f$-divergences introduced in \cite{capel2026geodesic}, which  are constructed by interpolating the relative modular operator and the commutant Radon-Nikodym derivative using a geodesic with parameter $t\in [0,1]$. For an invertible state $\rho$ and an operator convex function $f$, we compute the Hessian and
obtain an explicit formula for the induced
monotone quantum information metric  $g_{\rho,t}^{(f)}$. Furthermore, we also compare these metrics with the Petz-Hasegawa metric and find the meaning of the interpolation parameter $t$ in this new geometry.  We next show that the interpolation of relative modular operators $\Gamma_t$ in the reference purification of a state $\sigma$ defines a canonical finite binary experiment $(p_t,q_t)$, and introduce a log-likelihood cumulant function $\Psi_{\rho,\sigma}(t,s)$, recovering  the Nussbaum–Szko\l a distributions  at $t=0$ and the Matsumoto construction at $t=1$. Finally, using Busemann functions, we endow this statistical framework with a geometric meaning in the cone of positive operators.

\end{abstract}

\vspace{-0.8em}
\setcounter{tocdepth}{2}
{\small
\tableofcontents
}
\clearpage

\section{Introduction}

A basic structural problem in quantum information theory is how to interpolate between the standard and maximal quantum extensions of a classical $f$-divergence.  For invertible states $\rho$ and $\sigma$ and an operator-convex function $f$, the standard Petz $f$-divergence is generated by the relative modular operator
\begin{equation*}
 \Delta_{\rho\mid\sigma}=L_\rho R_{\sigma^{-1}},
 \qquad
 D_f(\rho\Vert\sigma)
 =\langle\Omega_\sigma,f(\Delta_{\rho\mid\sigma})\Omega_\sigma\rangle_{\mathrm{HS}},
\end{equation*}
whereas the maximal $f$-divergence is generated by the   commutant Radon--Nikodym derivative
\begin{equation*}
 K_1=\sigma^{-1/2}\rho\sigma^{-1/2},
 \qquad
 \widehat{D}_f(\rho\Vert\sigma)
 =\tr\!\left[\sigma f(K_1)\right]
\end{equation*}
\cite{matsumoto2018maximal,hiaiMosonyi2017,hiai2019maximal}.  The two constructions agree when $\rho$ and $\sigma$ commute, but in the noncommutative setting they provide genuinely different extensions \cite{hiaiMosonyi2017}.  For $f(x)=x\log x$, their values are, in particular,  the Umegaki and Belavkin--Staszewski relative entropies \cite{umegaki1962,belavkinStaszewski1982}.

There is no unique way of interpolating between these two families. The closest previous interpolation at the level of monotone additive relative entropies is due to Mosonyi, Bunth and Vrana \cite{mosonyiBunthVrana2024}.  Starting from a relative entropy $D^q$, they move the second state along the Kubo--Ando path $\sigma\#_\gamma\rho$ and do not modify the relative modular operator. 
For $D^q$ starting in the  Umegaki entropy, this gives an increasing family from Umegaki to Belavkin--Staszewski.  Related geometric and sharp R\'enyi constructions also single out the maximal endpoint through their  limits \cite{fangFawzi2021,berghSalzmannDatta2021}. The construction introduced in \cite{capel2026geodesic} is different: the pair $(\rho,\sigma)$ and the function  $f$ are kept fixed, while the quantum relative operator itself is transported along the affine-invariant geodesic
\begin{equation}\label{eq:relative-operator-geodesic-updated}\Gamma_t(\rho,\sigma):=L_{\rho^{1-t}}R_{K_t}
 =\Delta_{\rho\mid\sigma}\#_t\widehat\Delta_{\rho\mid\sigma}\, ,
 \qquad
 K_t:=\sigma^{-1/2}\rho^t\sigma^{-1/2}\, ,
 \qquad 0\le t\le1,
\end{equation}
where $\widehat\Delta_{\rho\mid\sigma}
 =R_{\sigma^{-1/2}\rho\sigma^{-1/2}}.$
The associated geodesic quantum $f$-divergences are
\begin{equation}\label{eq:intro-geodesic-divergence}
 D_f^t(\rho\Vert\sigma)
 :=\left\langle\Omega_\sigma,
 f(\Gamma_t)\Omega_\sigma\right\rangle_{\mathrm{HS}},
 \qquad \Omega_\sigma:=\sigma^{1/2}.
\end{equation}
 These divergences recover the standard and maximal constructions at $t=0$ and $t=1$, respectively, and satisfy data processing under quantum channels \cite{capel2026geodesic}.  The mathematical application of this relative-operator interpolation is the common object studied here.

In this work, we are concerned with the study of these geodesic $f$-divergences in terms of their geometry and statistics.  We show that their  Hessians   generate monotone Riemannian metrics, a theory that was  developed for quantum divergences through the work on quasi-entropies, noncommutative $\alpha$-divergences and monotone metrics by many authors including Petz, Hasegawa, Lesniewski, Ruskai, etc. \cite{petz1986quasi,hasegawa1993,hasegawa1997non,petzHasegawa1996,lesniewskiRuskai1999,petz1994canonical,petz1996monotone,hiai-2012}.  For the generator $f(x)=x\log(x)$, the standard endpoint yields the Bogoliubov--Kubo--Mori metric, whose origins lie in the Kubo--Mori canonical correlation and quantum statistical mechanics \cite{kubo1957,mori1965,petzToth1993,petz1994canonical}, and  the maximal endpoint yields the right-logarithmic-derivative metric familiar from multiparameter quantum estimation \cite{yuenLax1973,holevo2011,petz1996monotone}.  A known Hasegawa--Petz path connects the metric endpoints by varying the scalar generator while keeping the standard relative modular operator fixed \cite{hasegawa1993,hasegawa1997non,petzHasegawa1996,besenyei2012hasegawa}.  The first question addressed here is therefore what local information geometry is produced when, instead, the generator is fixed and the relative operator moves along $t\mapsto\Gamma_t$.

The two endpoints also have established statistical interpretations.  Classical binary inference is organized by the likelihood ratio and its log-moment-generating function, whose derivatives and variational transforms control relative entropy, information variance, Stein asymptotics, Chernoff information and large deviations \cite{chernoff1952,coverThomas2006}.  At $t=0$, the spectral measure of $\Delta_{\rho\mid\sigma}$ gives the Nussbaum--Szko\l a distributions, and
\begin{equation*}
 \log\left\langle\Omega_\sigma,
 \Delta_{\rho\mid\sigma}^{\,s}\Omega_\sigma\right\rangle_{\mathrm{HS}}
 =\log\tr\rho^s\sigma^{1-s}
\end{equation*}
is the relative-modular moment function used in quantum hypothesis testing and  statistics \cite{nussbaumSzkola2009,jaksicEtAl2012}.  

The corresponding exact classicalization of standard quantum $f$-divergences is made explicit in works like \cite{matsumoto2018maximal} or \cite{androulakisJohn2024}.  At $t=1$, the spectral decomposition of $\sigma^{-1/2}\rho\sigma^{-1/2}$ gives Matsumoto's optimal reverse test for the maximal $f$-divergence \cite{matsumoto2018maximal}, with explicit realizing distributions described also in \cite{lanierBeguinotRioul2025}.  The second question of the paper is whether every intermediate $\Gamma_t$ similarly determines a canonical binary experiment, and whether one surface can organize its complete likelihood calculus. We give an affirmative answer to this question.

Finally, the same construction has a natural global geometry.  The cone of strictly positive operators, equipped with the affine-invariant metric, is a nonpositively curved symmetric space whose geodesics are weighted operator geometric means \cite{corachPortaRecht1993,bhatiaHolbrook2006,bhatia2007}.  Busemann functions are the associated asymptotic height functions: their level sets are horospheres and their sublevel sets are horoballs \cite{bridsonHaefliger1999}.  Logarithmic expectations of the form
\begin{equation*}
 A\longmapsto\log\langle v,Av\rangle
\end{equation*}
are known to occur as rank-one Busemann, or Kempf--Ness, functions on the positive cone, and their Hessians have an expectation--variance interpretation \cite[Section 6]{hiraiNieuwboerWalter2026}.  Busemann functions have also appeared directly in statistical models through Poisson kernels and Fisher geometry \cite{itohShishido2008,itohSatoh2015}, while the Chernoff point is classically described through likelihood-ratio exponential families \cite{nielsen2022chernoff}.  The third question addresed here is whether these established geometric and statistical structures coincide for the particular rays generated by $\Gamma_t$.

\subsection*{Main results}

In this work we are concerned with the geometric and statistical applications of the geodesic divergences studied in \cite{capel2026geodesic}.  Figure \ref{fig:intro-three-contributions} summarizes the main results and  the organization of the paper. 

\begin{figure}[ht]
\centering
\begin{tikzpicture}[
  box/.style={draw=black,rounded corners=2pt,line width=0.8pt,align=center,inner sep=7pt,text=black},
  arr/.style={-{Latex[length=2.6mm]},draw=black,line width=0.9pt},
  lab/.style={font=\scriptsize,fill=white,inner sep=1.5pt,text=black}
]
\node[box,text width=7.0cm] (source) at (0,0) {
 \textbf{Geodesic quantum $f$-divergences}\\[1mm]
 $\Gamma_t=\Delta_{\rho\mid\sigma}\#_t\widehat\Delta_{\rho\mid\sigma}$,\qquad
 $D_f^t(\rho\Vert\sigma)=\langle\Omega_\sigma,f(\Gamma_t)\Omega_\sigma\rangle_{\mathrm{HS}}$
};
\node[box,text width=4.25cm] (metric) at (-5.0,-3.7) {
 \textbf{Local information\\ metric}\\[1mm]
 Hessian $g_{\rho,t}^{(f)}$\\
 representing function $m_{f,t}$
};
\node[box,text width=4.25cm] (stats) at (0,-3.7) {
 \textbf{Likelihood-statistics\\ surface}\\[1mm]
 canonical experiment $(p_t,q_t)$\\
 $\Psi_{\rho,\sigma}(t,s)$
};
\node[box,text width=4.25cm] (busemann) at (5.0,-3.7) {
 \textbf{Busemann and\\ horospherical geometry}\\[1mm]
 $b_\sigma(\Gamma_t^s)=\Psi_{\rho,\sigma}(t,s)$
};
\draw[arr] (source.south west) -- node[lab,sloped,above,pos=.38] {Hessian geometry} (metric.north);
\draw[arr] (source.south) -- node[lab,right] {spectral decomposition} (stats.north);
\draw[arr] (source.south east) -- node[lab,sloped,above,pos=.34] {affine-invariant geometry} (busemann.north);
\end{tikzpicture}
\caption{The three  applications of the  geodesic $f$-divergences.  Section \ref{sec:Section1} studies their local Hessian geometry, Section \ref{sec:Section_2} its  statistics, and Section \ref{sec:horospherical-likelihood-geometry} its Busemann and horospherical geometry.}
\label{fig:intro-three-contributions}
\end{figure}
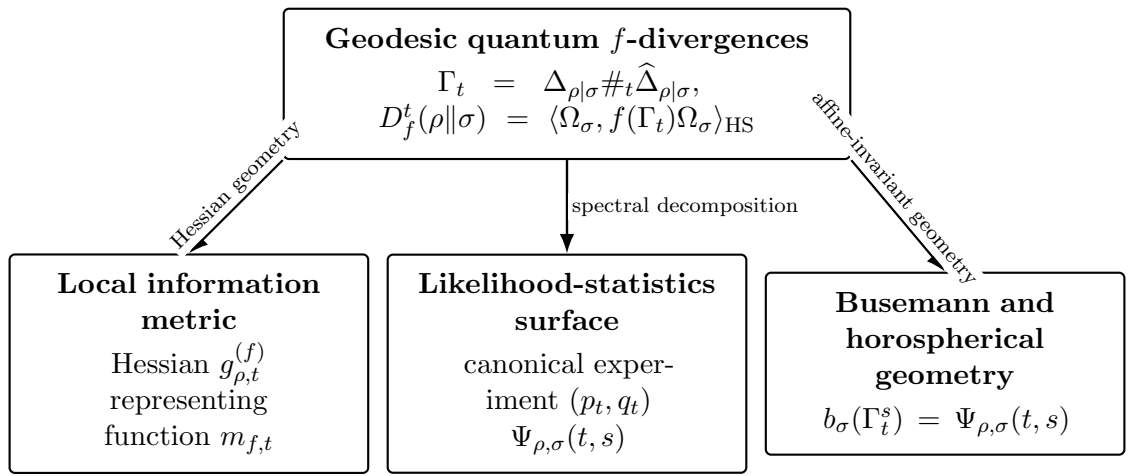

\paragraph{1. Local information metric.}

Section 2 computes the Hessian of $D_t^f$ 
for every non-affine operator-convex generator $f$. If an invertible quantum state $\rho$ has spectral decomposition $\rho=\sum_i p_i \vert v_i \rangle \langle v_i\vert $, we show in Theorem \ref{theo:Hessian_metric_general} that $D_t^f$
 induces a metric of the form
\begin{equation*}
 g_{\rho,t}^{(f)}(X,Y)
 =\sum_{i,j}
 \frac{f''(1)}{p_jm_{f,t}(p_i/p_j)}\,
 \overline{X}_{ij}Y_{ij},
\end{equation*}
with an explicit normalized symmetric representing function
\begin{equation*}
 m_{f,t}(r)
 =\frac{f''(1)(r^{1-t}-1)^2}
 {\widetilde f(r^{1-t})+r^{1-2t}\widetilde f(r^{t-1})},
\end{equation*}
where $\widetilde f(x)=f(x)-f(1)-f'(1)(x-1)$. In addition, we also show in Proposition \ref{prop:monotone_metric} that this is a monotone metric. This result extends the work of \cite[Section 2.3]{hiai-2012} to every $t\in [0,1]$. Using continuity of the interpolation in $t$, we recover Matsumoto’s result  that the Hessian at $t=1$ is universally the RLD metric and depends only on $f''(1)$ \cite{matsumoto2018maximal}. In addition,  for $f(x)=x\log x$, the metrics form a monotone and convex path from BKM to RLD.  The same parameter that locates $\Gamma_t$ on the global affine-invariant geodesic can be recovered from the local metric alone through the boundary index
\begin{equation*}
 \lim_{r\downarrow0}\frac{\log m_t(r)}{\log r}=t.
\end{equation*}
The comparison with the Hasegawa--Petz interpolation \cite{petzHasegawa1996,hasegawa1997non} shows that while their interpolation fixes the relative modular operator and changes the operator convex function $f_t$, the interpolation presented in this work fixes the function and interpolates the relative modular operator.

\paragraph{2. Canonical likelihood statistics.}
Section \ref{sec:Section_2} proves the quantum-to-classical step for statistics.  If we write the spectral decomposition $\Gamma_t=\sum_z\lambda_{t,z}E_{t,z}$, 
then the spectral measurement on $\Omega_\sigma$ and on
$ \Omega_{\rho,t}:=\Gamma_t^{1/2}\Omega_\sigma$
produces distributions $q_t$ and $p_t$, respectively; moreover, $\Omega_{\rho,t}$ purifies $\rho$ and
\begin{equation*}
 \frac{p_t(z)}{q_t(z)}=\lambda_{t,z}.
\end{equation*}
Thus the spectrum of the quantum relative operator is exactly the likelihood ratio of a canonical finite classical experiment, and
\begin{equation*}
 D_f^t(\rho\Vert\sigma)=D_f(p_t\Vert q_t).
\end{equation*}
 We also introduce the two-parameter surface
\begin{equation}\label{eq:intro-likelihood-surface}
 \Psi_{\rho,\sigma}(t,s)
 :=\log\left\langle\Omega_\sigma,
 \Gamma_t^s\Omega_\sigma\right\rangle_{\mathrm{HS}}
 =\log\sum_zp_t(z)^sq_t(z)^{1-s}\, ,
\end{equation}
which corresponds to the  log-likelihood cumulant of $(p_t,q_t)$.  Its endpoint slope and curvature give the relative entropy and information variance, its minimum gives the negative Chernoff information, and its Legendre transform gives the likelihood large-deviation tradeoff.  This follows by the well-known classical results.

\paragraph{3. Busemann and horospherical geometry.}
Section \ref{sec:horospherical-likelihood-geometry} provides a geometric interpretation of the likelihood surface in the affine-invariant positive
cone.  Logarithmic quadratic forms of the type $A\mapsto\log\langle v,Av\rangle$ are known to admit a Busemann interpretation in the positive-definite cone; see, for example, \cite[Section 6.1]{hiraiNieuwboerWalter2026}.  Specializing to $v=\Omega_\sigma$, we use the normalized Busemann function
\begin{equation*}
 b_\sigma(A)=\log\langle\Omega_\sigma,A\Omega_\sigma\rangle_{\mathrm{HS}}.
\end{equation*}
In our setting, the two main special properties are
\begin{equation*}
 b_\sigma(\Gamma_t)=0,
 \qquad
 b_\sigma(\Gamma_t^s)=\Psi_{\rho,\sigma}(t,s).
\end{equation*}
Hence the entire $t$-geodesic lies in one normalization horosphere, while for fixed $t$ the radial geodesic $s\mapsto\Gamma_t^s$ moves through the horospherical levels.  The tilted mean log-likelihood becomes its signed normal velocity, the likelihood variance becomes the Hessian of the Busemann height along the radial curve, and the Chernoff point is the unique point of maximal depth in the normalization horoball.  In particular, the Chernoff information is the distance from this point to the normalization horosphere, see Figure \ref{fig:horosphere}.  

\subsection*{Notation and preliminaries}

Throughout this work, $\mathcal H$ will denote a finite dimensional Hilbert space and $\mathcal B_2(\mathcal H)$ denotes the Hilbert--Schmidt space with inner product
\begin{equation*}
 \langle X,Y\rangle_{\mathrm{HS}}:=\operatorname{tr}(X^*Y).
\end{equation*}
For $A,B\in\mathcal B(\mathcal H)$, the operators $L_A$ and $R_B$ act on $\mathcal B_2(\mathcal H)$ by
\begin{equation*}
 L_A(X)=AX,
 \qquad
 R_B(X)=XB.
\end{equation*}
The trace on $\mathcal H$ is denoted by $\operatorname{tr}$, while $\operatorname{tr}_{\mathcal B_2(\mathcal H)}$ denotes the operator trace on the Hilbert--Schmidt space when a distinction is needed.  All logarithms are natural.

Let $\mathbb P(\mathcal{H})$ denote the cone of strictly positive self-adjoint operators on $\mathcal{H}$.  Its tangent space at every point is the real vector space of self-adjoint operators, and the affine-invariant metric is
\begin{equation}\label{eq:metric}
 g_A(X,Y)
 :=\operatorname{tr}
 (A^{-1}XA^{-1}Y)
 =\left\langle A^{-1/2}XA^{-1/2},
 A^{-1/2}YA^{-1/2}\right\rangle_{\mathrm{HS}}.
\end{equation}
The cone is convex and therefore simply connected.  With this metric it is complete and has nonpositive sectional curvature, hence is a finite-dimensional Hadamard manifold; see \cite[Chapter 6]{bhatia2007}, together with \cite{bhatia2003exponential,bhatiaHolbrook2006,corachPortaRecht1993}.  Its distance and unique geodesic are
\begin{equation*}
 d_2(A,B)
 =\left\|\log(A^{-1/2}BA^{-1/2})\right\|_{\mathrm{HS}},
\end{equation*}
and
\begin{equation*}
 A\#_sB
 =A^{1/2}(A^{-1/2}BA^{-1/2})^sA^{1/2},
 \qquad 0\le s\le1.
\end{equation*}

\section{Monotone metrics and geodesic $f$-divergences}\label{sec:Section1}

\subsection{The Hessian  metric}

The set of invertible density matrices is an open manifold relative to the full manifold of states, and its tangent space at $\rho$ is given by \cite{petz2008book}
\begin{equation*}
 T_\rho\cS(\cH)
 =\{X\in\cB(\cH):X=X^*,\ \tr X=0\}.
\end{equation*}
For a quantum state $\rho$ and traceless self-adjoint $X,Y \in T_{\rho}(\mathcal{S}(\mathcal{H}))$, it is possible to introduce a real bilinear form using the Hessian of the geodesic divergences, i.e.
\begin{equation}\label{eq:Definition_metric_Hessian}
    g_{\rho,t}^{(f)}(X,Y)
    :=-\left.\frac{\partial^2}{\partial s\,\partial u}
    D_f^t(\rho+sX\Vert \rho+uY)\right|_{s=u=0},
\end{equation}
where $f:(0,\infty) \to  \mathbb{R}$ is a non-affine operator convex function. This is the usual construction of a metric from a smooth divergence in information geometry \cite{eguchi1992geometry,amariNagaoka2000}; in the quantum setting, Hessians of quasi-entropies and relative entropies were studied in \cite{hasegawa1993,hasegawa1997non,lesniewskiRuskai1999,petz1986quasi,petz1994canonical,petzHasegawa1996}. Notice that the affine part of $f$ does not contribute to the Hessian and, for this reason, it will be useful to denote
\begin{equation*}
    \widetilde{f}(x)=f(x)-f(1)-f'(1)(x-1)\, ,
\end{equation*}
so that $\widetilde{f}(1)=0$, $\widetilde{f}'(1)=0$ and $\widetilde{f}''(1)=f''(1)$.

  Now, fix a spectral decomposition
\begin{equation*}
    \rho=\sum_{i=1}^d p_i E_{ii},
    \qquad E_{ij}:=|v_i\rangle\langle v_j|.
\end{equation*}
For $i<j$, set
\begin{equation*}
    A_{ij}:=E_{ij}+E_{ji},
     \qquad
B_{ij}:=\mathrm{i}(E_{ij}-E_{ji}).
\end{equation*}
   Then  we can decompose $T_{\rho}(\mathcal{S}(\mathcal{H}))\cong \mathcal{D}_0\oplus \text{span}_{\mathbb{R}}\{ A_{ij},B_{ij}\}$, where  $\mathcal{D}_0$ corresponds to real diagonal matrices with trace zero.

Although $g_{\rho,t}^{(f)}$ is defined on the real vector space of self-adjoint tangent vectors, we use the same symbol for its canonical sesquilinear extension to the complexification;  this extension is conjugate-linear in the first entry and linear in the second. 

\begin{remark}
    Although $g_{\rho,t}^{(f)}$ is introduced above on the manifold of faithful
density matrices, the defining expression for $D_f^t(A\Vert B)$ is well-defined
for arbitrary strictly positive matrices $A,B>0$. Accordingly, whenever the
base point $A>0$ is not normalized, we use the same notation for the Hessian
form on the positive cone,
\begin{equation*}
 g_{A,t}^{(f)}(X,Y)
 :=
 -\left.
 \frac{\partial^2}{\partial s\,\partial u}
 D_f^t(A+sX\Vert A+uY)
 \right|_{s=u=0},
\end{equation*}
where $X,Y$ are self-adjoint matrices and $s,u$ are sufficiently small so that
$A+sX$ and $A+uY$ remain strictly positive.
\end{remark}

The next lemma records the orthogonality relations among the above sectors.

\begin{lemma}\label{lemma:Orthogonality_Relations}

Let $\rho=\sum_i p_iE_{ii}>0$, and  let  $D$ be  a traceless diagonal self-adjoint matrix in the eigenbasis of $\rho$. Then, for $i<j$ and $k<l$,
\begin{subequations}
\begin{align}
     g_{\rho,t}^{(f)}(D,A_{ij})
&=g_{\rho,t}^{(f)}(D,B_{ij})=0,\label{eq:Orthogonality_D}\\
         g_{\rho,t}^{(f)}(A_{ij},A_{kl})
   &=g_{\rho,t}^{(f)}(B_{ij},B_{kl})
    =g_{\rho,t}^{(f)}(A_{ij},B_{kl})=0,
    \qquad (i,j)\neq(k,l),\label{eq:orthogonality_pairs}\\
g_{\rho,t}^{(f)}(A_{ij},B_{ij})&=0,\label{eq:Orthogonality_A_B}\\
         g_{\rho,t}^{(f)}(A_{ij},A_{ij})
   &=g_{\rho,t}^{(f)}(B_{ij},B_{ij}).\label{eq:Same_norm_AB}
\end{align}
\end{subequations}
\end{lemma}
\begin{proof}
The proof is given in   Lemma \ref{lemma:proof-orthogonality-relations} in the Appendix.
\end{proof}

    The previous lemma shows that the metric $g_{\rho,t}^{(f)}$ preserves the orthogonality relations of  the  decomposition  of the tangent space into  the traceless diagonal  sector and the two-dimensional real sectors spanned by $A_{ij}$ and $B_{ij}$. It is therefore enough to compute the  metric on one vector $A_{ij}$ in each off-diagonal sector. The next lemma drastically simplifies this computation: the corresponding coefficient depends on $(p_i,p_j)$ only through the scale $p_j^{-1}$ and the ratio $p_i/p_j$.

\begin{lemma}\label{lemma:Metric_In_Selfadjoint_Basis}
     Let $\rho\in \cS(\cH)$ be an invertible quantum state with eigenvalues  $(p_i)_i$ and corresponding eigenvectors $(v_i)_i$. Let $E_{ij}=\vert v_i \rangle \langle v_j \vert$ and define the traceless self-adjoint matrix $A_{ij}=E_{ij}+E_{ji}$. Then, for any operator convex function $f:(0,\infty)\to \mathbb{R}$ with $f(1)=0$,
\begin{equation}\label{eq:Equation_Statement_Lemma_Metric_Basis}
    g_{\rho,t}^{(f)}(A_{ij},A_{ij})=\frac{1}{p_j}g_{P_r,t}^{(f)}(A,A)\, ,
\end{equation}
where, for $r=p_i/p_j$,
\begin{equation*}
         P_r=\begin{pmatrix}
            r & 0\\
             0 & 1
         \end{pmatrix},
         \qquad
         A=\begin{pmatrix}
            0 & 1\\
            1 & 0
         \end{pmatrix}.
\end{equation*}
\end{lemma}
\begin{proof}
The proof is given in Lemma \ref{lemma:proof-two-level-reduction} in the Appendix.
\end{proof}

We are now ready to state the main result of this section. Before computing the Hessian explicitly, recall that monotone quantum information metrics admit a complete characterization in terms of normalized symmetric operator-monotone functions.  More precisely, if $\rho=\sum_i p_i\vert v_i\rangle\langle v_i\vert$, every monotone metric can be written in the form
\begin{equation}
g_\rho(X,Y)
=
\sum_{i,j}
\frac{\overline{X}_{ij}Y_{ij}}
{p_jm(p_i/p_j)},
\end{equation}
where $m:(0,\infty)\to \mathbb{R}$ is operator monotone, $m(1)=1$, and $m(r)=rm(r^{-1})$ \cite{morozovaChentsov1991,petz1996monotone,lesniewskiRuskai1999}.  Notice that the function $m$ appearing in this classification is different from the operator-convex function $f$ defining the divergence $D_f^t$. Furthermore, for the case of $f$-divergences the explicit formula was shown in \cite[eq. (14)]{hiai-2012} using the results obtained in  \cite{lesniewskiRuskai1999}. The next theorem computes the metric for the interpolated divergences recovering the expression in \cite{hiai-2012} for the case $t=0$.

\begin{theorem}\label{theo:Hessian_metric_general}
     Let $f:(0,\infty)\to \mathbb{R}$ be a non-affine operator convex function and define $\kappa_f=f''(1)>0$. For the invertible quantum state $\rho  \in \mathcal{S}(\mathcal{H})$, write its spectral decomposition $\rho=\sum_i p_i \vert v_i \rangle \langle  v_i  \vert$, where $(v_i)_i$ is an orthonormal basis of eigenvectors of $\rho$. Then, for $X,Y\in T_\rho\cS(\cH)$, the Hessian of $D_f^t$ has the form
    \begin{equation}\label{eq:Formula_metric_general}
         g_{\rho,t}^{(f)}(X,Y)
        =\sum_{i,j}\frac{\kappa_f}{p_j m_{f,t}(p_i/p_j)}
        \overline{X}_{ij}Y_{ij},
    \end{equation}
where, for $0\leq t <1$,
\begin{equation}\label{eq:Definition_mft}
    m_{f,t}(r)
    =\frac{\kappa_f(r^{1-t}-1)^2}
     {\widetilde{f}(r^{1-t})+r^{1-2t}\widetilde{f}(r^{t-1})}\, ,
    \qquad m_{f,t}(1)=1\, .
\end{equation}
The function $m_{f,t}$ further satisfies
\begin{equation}\label{eq:Symmetry_m}
 m_{f,t}(r)=r\,m_{f,t}(r^{-1}),
\end{equation}
and all removable singularities are filled in by continuity.
 \end{theorem}

\begin{proof}

    The strict inequality $\kappa_f>0$ follows from the canonical integral representation of operator-convex functions on $(0,\infty)$ \cite{franz2014higher,lesniewskiRuskai1999}. Namely, one can write
\begin{equation}\label{eq:integral_rep}
f(x)=a+b(x-1)+c(x-1)^2
        +\int_{[0,\infty)}\frac{(x-1)^2}{x+s}\,\dd\mu(s),
\end{equation}
  with $c\geq0$ and $\mu$ a positive measure. Hence
\begin{equation*}
f''(1)=2c+2\int_{[0,\infty)}\frac{1}{1+s}\,\dd\mu(s).
\end{equation*}
  If this number vanished, then $c=0$  and $\mu=0$, so $f$ would be affine. Thus non-affinity implies $\kappa_f=f''(1)>0$.

Because the affine part of $f$ does not contribute to the mixed Hessian, we may apply all subsequent calculations to the centered function $\widetilde f$. In particular, the hypotheses $\widetilde  f(1)=\widetilde f'(1)=0$  required in Lemma \ref{lemma:Taylor} are satisfied.

In Lemma \ref{lemma:Metric_In_Selfadjoint_Basis} we showed  that, for the diagonal $2\times2$ matrix
\begin{equation*}
 P=\begin{pmatrix}r&0\\0&1\end{pmatrix}
\end{equation*}
and the rotation $P_{\theta}=R_{\theta}PR_{\theta}^*$, one  has
\begin{equation*}
 P_0'=(r-1)A,
 \qquad
 A=\begin{pmatrix}0&1\\1&0\end{pmatrix}.
\end{equation*}
By Lemma \ref{lemma:Metric_In_Selfadjoint_Basis}, our  goal reduces to computing $g_{P,t}^{(f)}(A,A)$. Since the metric is real bilinear on the  self-adjoint  tangent space,
\begin{equation*}
         g_{P,t}^{(f)}(P_0',P_0')=(r-1)^2g_{P,t}^{(f)}(A,A)\, .
\end{equation*}
As explained in Lemma \ref{lemma:Metric_In_Selfadjoint_Basis},
\begin{equation}\label{eq:second_derivative_rotation}
        (r-1)^2g_{P,t}^{(f)}(A,A)
        =\left. \frac{\dd^2}{\dd \theta^2}
        D_f^t(P\Vert P_{\theta})\right|_{\theta=0}.
\end{equation}
Thus it is enough to determine the coefficient  of $\theta^2$ in the expansion of $D_f^t(P\Vert P_{\theta})$.

Set
\begin{equation*}
 Z_{\theta}:=P_{\theta}^{-1/2}P^tP_{\theta}^{-1/2}\, ,
\end{equation*}
so that we can write
\begin{equation*}
    D_f^t(P\Vert P_{\theta})
    =\langle P_{\theta}^{1/2},
      \widetilde f(L_{P^{1-t}}R_{Z_{\theta}})
      P_{\theta}^{1/2}\rangle_{\mathrm{HS}}\, .
\end{equation*}
Define now the row subspaces
\begin{equation*}
 \cW_1=\operatorname{span}\{E_{11},E_{12}\},
 \qquad
 \cW_2=\operatorname{span}\{E_{21},E_{22}\}.
\end{equation*}
Since left multiplication by $P^{1-t}=\operatorname{diag}(r^{1-t},1)$ multiplies the first row by  $r^{1-t}$ and leaves  the second  row unchanged, while right multiplication by $Z_\theta$ acts separately on each row, $\cW_1$ and $\cW_2$ are invariant under $L_{P^{1-t}}R_{Z_\theta}$. This operator  is self-adjoint on the Hilbert--Schmidt space, so these invariant  subspaces are reducing. Under the natural identification of each row with $\mathbb C^2$, its restrictions are represented by $r^{1-t}Z_\theta$ and $Z_\theta$, respectively.  Therefore, if $a_\theta$ and $b_\theta$ denote the first and second rows of $P_\theta^{1/2}$,  then functional calculus gives the exact decomposition
\begin{equation*}
     D_f^t(P\Vert P_{\theta})
    =\langle a_{\theta},\widetilde f(r^{1-t}Z_{\theta})a_{\theta}\rangle
      +\langle b_{\theta},\widetilde f(Z_{\theta})b_{\theta}\rangle\, .
\end{equation*}

By Lemma \ref{lemma:Taylor},
\begin{equation}\label{eq:Taylor_in_main_theorem}
 D_f^t(P\Vert P_\theta)
 =(r-1)^2
 \frac{\widetilde f(r^{1-t})
       +r^{1-2t}\widetilde f(r^{t-1})}
      {(r^{1-t}-1)^2}\,\theta^2
 +o(\theta^2).
\end{equation}
Combining \eqref{eq:second_derivative_rotation} and \eqref{eq:Taylor_in_main_theorem} yields
\begin{equation}\label{eq:value_of_the_metric_f}
    g_{P,t}^{(f)}(A,A)
    =2\frac{\widetilde f(r^{1-t})
       +r^{1-2t}\widetilde f(r^{t-1})}
      {(r^{1-t}-1)^2}
    =\frac{2\kappa_f}{m_{f,t}(r)}.
\end{equation}
Notice that since the function is operator convex and non-affine $\widetilde f(x)>0$  for $x\neq 1$, so $m_{f,t}(r)$ is well-defined. This is a consequence of its integral representation \eqref{eq:integral_rep}.

  We next verify the normalization and symmetry of $m_{f,t}$. Put $q=r^{1-t}$. As $r\to1$, Taylor's  expansion gives
\begin{equation*}
        \widetilde f(q)=\frac{\kappa_f}{2}(q-1)^2+o((q-1)^2)
\end{equation*}
  and
\begin{equation*}
\widetilde f(q^{-1})
         =\frac{\kappa_f}{2}(q^{-1}-1)^2+o((q^{-1}-1)^2)
   =\frac{\kappa_f}{2q^2}(q-1)^2+o((q-1)^2).
\end{equation*}
Since $r^{1-2t}=q^2/r$, the denominator in \eqref{eq:Definition_mft} is
\begin{equation*}
        \frac{\kappa_f}{2}(q-1)^2
   +\frac{\kappa_f}{2r}(q-1)^2
    +o((q-1)^2)
    =\frac{\kappa_f}{2}\left(1+\frac{1}{r}\right)(q-1)^2+o((q-1)^2).
\end{equation*}
Hence $m_{f,t}(r)\to1$, proving $m_{f,t}(1)=1$.

    For the symmetry,  again write $q=r^{1-t}$. Then
\begin{align*}
m_{f,t}(r^{-1})
        &=\frac{\kappa_f(q^{-1}-1)^2}
   {\widetilde f(q^{-1})+(r^{-1})^{1-2t}\widetilde  f(q)}\\
    &=\frac{\kappa_f(q-1)^2}
    {q^2\widetilde f(q^{-1})+r\widetilde f(q)},
\end{align*}
    where  we used $(r^{-1})^{1-2t}=r^{2t-1}=r/q^2$. Multiplying by $r$ and dividing numerator and denominator by $r$ gives
\begin{equation*}
    r\,m_{f,t}(r^{-1})
     =\frac{\kappa_f(q-1)^2}
{\widetilde f(q)+(q^2/r)\widetilde f(q^{-1})}
        =m_{f,t}(r),
\end{equation*}
which proves \eqref{eq:Symmetry_m}.

  We now return to the general state $\rho$. Lemma \ref{lemma:Metric_In_Selfadjoint_Basis} and \eqref{eq:value_of_the_metric_f} give, for $i<j$,
\begin{equation}\label{eq:Metric_Aij}
    g_{\rho,t}^{(f)}(A_{ij},A_{ij})
    =\frac{2\kappa_f}
    {p_jm_{f,t}(p_i/p_j)}.
\end{equation}
By Lemma \ref{lemma:Orthogonality_Relations},
\begin{equation}\label{eq:Metric_Bij}
    g_{\rho,t}^{(f)}(B_{ij},B_{ij})
    =\frac{2\kappa_f}
{p_jm_{f,t}(p_i/p_j)},
        \qquad
   g_{\rho,t}^{(f)}(A_{ij},B_{ij})=0.
\end{equation}

    It remains to compute the diagonal sector. Let
\begin{equation*}
   D_X=\sum_i x_iE_{ii},
     \qquad
     D_Y=\sum_i y_iE_{ii},
\qquad
          \sum_i x_i=\sum_i y_i=0.
\end{equation*}
    All operators involved commute, and therefore the geodesic divergence reduces to the classical $f$-divergence, independently  of $t$:
\begin{equation*}
D_f^t(\rho+sD_X\Vert\rho+uD_Y)
        =\sum_i(p_i+uy_i)
   \widetilde f\!\left(\frac{p_i+sx_i}{p_i+uy_i}\right).
\end{equation*}
For the $i$th summand, differentiation first in $s$ gives
\begin{equation*}
        \frac{\partial}{\partial s}
   \left[(p_i+uy_i)
    \widetilde f\!\left(\frac{p_i+sx_i}{p_i+uy_i}\right)\right]
    =x_i\widetilde f'\!\left(\frac{p_i+sx_i}{p_i+uy_i}\right).
\end{equation*}
  Differentiating in $u$ and setting $s=u=0$ gives
\begin{equation*}
    \left.\frac{\partial^2}{\partial u\,\partial s}\right|_{s=u=0}
    (p_i+uy_i)
\widetilde f\!\left(\frac{p_i+sx_i}{p_i+uy_i}\right)
        =-\frac{\kappa_f x_i y_i}{p_i}.
\end{equation*}
    The minus sign in the definition \eqref{eq:Definition_metric_Hessian} therefore yields
\begin{equation}\label{eq:Diagonal_metric}
g_{\rho,t}^{(f)}(D_X,D_Y)
          =\kappa_f\sum_i\frac{x_i y_i}{p_i}.
\end{equation}

Finally, write
\begin{equation*}
        X=D_X+\sum_{i<j}(a_{ij}A_{ij}+b_{ij}B_{ij}),
   \qquad
    Y=D_Y+\sum_{i<j}(c_{ij}A_{ij}+d_{ij}B_{ij}),
 \end{equation*}
    where all coefficients are real. Using  Lemma \ref{lemma:Orthogonality_Relations}, \eqref{eq:Metric_Aij}, \eqref{eq:Metric_Bij}, and \eqref{eq:Diagonal_metric}, we obtain
\begin{align}\label{eq:Metric_real_basis}
   g_{\rho,t}^{(f)}(X,Y)
    &=\kappa_f\sum_i\frac{x_i y_i}{p_i}
    +\sum_{i<j}
\frac{2\kappa_f(a_{ij}c_{ij}+b_{ij}d_{ij})}
         {p_jm_{f,t}(p_i/p_j)}.
\end{align}
    For $i<j$,
\begin{equation*}
X_{ij}=a_{ij}+\mathrm{i}b_{ij},
         \qquad
   Y_{ij}=c_{ij}+\mathrm{i}d_{ij},
\end{equation*}
so
\begin{equation*}
        a_{ij}c_{ij}+b_{ij}d_{ij}
   =\operatorname{Re}(\overline{X_{ij}}Y_{ij}).
\end{equation*}
    Since $X_{ji}=\overline{X_{ij}}$ and $Y_{ji}=\overline{Y_{ij}}$, the sum of the two ordered terms associated with the unordered pair $\{i,j\}$ is
\begin{equation}
         \frac{\kappa_f\overline{X_{ij}}Y_{ij}}
   {p_jm_{f,t}(p_i/p_j)}
    +\frac{\kappa_f\overline{X_{ji}}Y_{ji}}
    {p_im_{f,t}(p_j/p_i)}=
        \frac{2\kappa_f\operatorname{Re}(\overline{X_{ij}}Y_{ij})}
   {p_jm_{f,t}(p_i/p_j)}.
\end{equation}
The diagonal terms in \eqref{eq:Formula_metric_general} agree with \eqref{eq:Diagonal_metric} because $m_{f,t}(1)=1$. Thus \eqref{eq:Metric_real_basis} is exactly the spectral formula \eqref{eq:Formula_metric_general}.

\end{proof}

The next result further shows that  $\kappa_f^{-1}g_{\rho,t}^{(f)}$ is a normalized
monotone quantum metric.  The Morozova--Chentsov--Petz classification
\cite{morozovaChentsov1991,petz1996monotone,lesniewskiRuskai1999} then implies
that its normalized symmetric representing function $m_{f,t}$ is operator
monotone.

\begin{proposition}\label{prop:monotone_metric}
    Under the assumptions of Theorem \ref{theo:Hessian_metric_general}, $\kappa_f^{-1}g_{\rho,t}^{(f)}$ is a monotone quantum information metric. More precisely, let
$\Phi:\cB(\cH)\to\cB(\mathcal K)$ be a quantum channel and put
\begin{equation*}
 P:=\supp\Phi(\rho).
\end{equation*}
After regarding $\Phi(\rho)$ and $\Phi(X)$ as operators on the support space
$P\mathcal K$, one has
\begin{equation}\label{eq:metric-DPI-support-corner}
 g_{\Phi(\rho),t}^{(f)}\!\left(\Phi(X),\Phi(X)\right)
 \leq g_{\rho,t}^{(f)}(X,X).
\end{equation}
Here the metric on the left is computed in the corner algebra
$P\cB(\mathcal K)P$, in which $\Phi(\rho)$ is invertible. In particular,
$m_{f,t}$ is a normalized symmetric operator-monotone function.
\end{proposition}
\begin{proof}
    We have to prove the monotonicity assertion for an arbitrary quantum channel
$\Phi:\cB(\cH)\to\cB(\mathcal K)$.  First notice that  even though
$\rho$ is invertible, the output state $\Phi(\rho)$ need not be invertible on the whole
output space $\mathcal K$.  We therefore compress the output algebra to the support of
$\Phi(\rho)$.

Put $P:=\supp\Phi(\rho)$.
We claim that the whole range of $\Phi$ is contained in the corner
$P\cB(\mathcal K)P$.  Indeed, if $A\geq0$, then, since $\rho>0$,
\begin{equation*}
 A\leq c_A\rho,
 \qquad
 c_A:=\left\|\rho^{-1/2}A\rho^{-1/2}\right\|_\infty<\infty.
\end{equation*}
Positivity of $\Phi$ gives
\begin{equation*}
 0\leq\Phi(A)\leq c_A\Phi(\rho).
\end{equation*}
Consequently,
\begin{equation*}
 \supp\Phi(A)\leq\supp\Phi(\rho)=P,
\end{equation*}
because every vector in the kernel of $\Phi(\rho)$ also belongs to the kernel of
$\Phi(A)$.  By linearity, it follows that
\begin{equation}\label{eq:channel-range-support-corner}
 \Phi(B)=P\Phi(B)P
 \qquad\text{for every }B\in\cB(\cH).
\end{equation}
Thus the same map can be regarded as a channel
\begin{equation*}
 \Phi_P:\cB(\cH)\longrightarrow\cB(P\mathcal K),
 \qquad
 \Phi_P(B):=\Phi(B)|_{P\mathcal K}.
\end{equation*}
It is completely positive and trace preserving, and $\Phi_P(\rho)$ is invertible on
$P\mathcal K$ by the definition of $P$.  Moreover, for every tangent vector
$X\in T_\rho\cS(\cH)$,
\begin{equation*}
 \Phi_P(X)=\Phi(X)\in T_{\Phi_P(\rho)}\cS(P\mathcal K).
\end{equation*}
In what follows, we suppress the subscript $P$ and understand all output
quantities in this support corner.  If $\Phi(\rho)$ is already invertible on $\mathcal K$,
then $P=I_\mathcal K$ and no compression is needed.

Since affine terms do not contribute to the Hessian, we may work with the centered
generator $\widetilde f$.  For all sufficiently small real $s$, both
$\rho+sX$ and $\Phi_P(\rho)+s\Phi_P(X)$ are invertible in their respective
algebras.  Data processing \cite[Theorem~6.2]{capel2026geodesic} applied to the
compressed channel gives
\begin{equation*}
 D_{\widetilde f}^t\!\left(\Phi_P(\rho)\middle\Vert
 \Phi_P(\rho)+s\Phi_P(X)\right)
 \leq
 D_{\widetilde f}^t(\rho\Vert\rho+sX).
\end{equation*}
Define
\begin{equation*}
 h_X(s):=
 D_{\widetilde f}^t(\rho\Vert\rho+sX)
 -D_{\widetilde f}^t\!\left(\Phi_P(\rho)\middle\Vert
 \Phi_P(\rho)+s\Phi_P(X)\right).
\end{equation*}
Then $h_X(s)\geq0$ near $s=0$ and $h_X(0)=0$. Hence $s=0$ is a local
minimum of $h_X$, so $h_X'(0)=0$ and $h_X''(0)\geq0$. By the standard
second-variation identity for a smooth divergence, see e.g. \cite[pp. 632--633]{eguchi1992geometry}, \cite[Proposition 2.7]{matsuzoe2025divergence}, or \cite[Eq. (4)]{ay2015novel} ,
\begin{equation*}
 \left.\frac{\dd^2}{\dd s^2}
 D_{\widetilde f}^t(\rho\Vert\rho+sX)\right|_{s=0}
 =g_{\rho,t}^{(f)}(X,X),
\end{equation*}
and the same identity in the support algebra $\cB(P\mathcal K)$ gives
\begin{equation*}
 h_X''(0)
 =g_{\rho,t}^{(f)}(X,X)
 -g_{\Phi(\rho),t}^{(f)}\!\left(\Phi(X),\Phi(X)\right).
\end{equation*}
Here and below, the metric at $\Phi(\rho)$ is understood in the corner
$P\cB(\mathcal K)P$. Therefore
\begin{equation}\label{eq:metric-DPI-quadratic}
 g_{\Phi(\rho),t}^{(f)}\!\left(\Phi(X),\Phi(X)\right)
 \leq g_{\rho,t}^{(f)}(X,X),
\end{equation}
which is exactly \eqref{eq:metric-DPI-support-corner}.

\end{proof}

\subsection{Interpolation between the BKM and the RLD metric}\label{sec:Interpolation_BKM_RLD}

The choice of the function $f(x)=x\log x$  plays  a fundamental role in quantum information theory. At the two endpoints it recovers, respectively, Petz's standard f-divergence and the maximal quantum $f$-divergence \cite{petz1986quasi,matsumoto2018maximal,hiaiMosonyi2017,hiai2019maximal}. We write
\begin{equation*}
       D^t(\rho\Vert\sigma):=D_{x\log x}^t(\rho\Vert\sigma),\qquad
       g_{\rho,t}:=g_{\rho,t}^{(x\log x)}.
\end{equation*}
At the initial point $t=0$, the relative operator is $\Delta_{\rho\mid\sigma}=L_\rho R_{\sigma^{-1}}$. Since left and  right multiplication commute, $ 
 \log\Delta_{\rho\mid\sigma}=L_{\log\rho}-R_{\log\sigma},$ and
a direct computation gives
\begin{equation*}
        D^0(\rho\Vert\sigma)
       =\tr\bigl[\rho(\log\rho-\log\sigma)\bigr]
        =:D(\rho\Vert\sigma),
\end{equation*}
which is the Umegaki relative entropy \cite{umegaki1962}. At the  endpoint $t=1$, \eqref{eq:relative-operator-geodesic-updated} gives $R_{K_1}$ 
where $K_1=\sigma^{-1/2}\rho\sigma^{-1/2}$ and therefore
\begin{equation}\label{eq:BS_endpoint_new}
       D^1(\rho\Vert\sigma)
       =\tr\bigl[\sigma K_1\log K_1\bigr]
        =\tr\left[\rho\log\left(\rho^{1/2}\sigma^{-1}\rho^{1/2}\right)\right]
       =:\widehat D(\rho\Vert\sigma),
\end{equation}
which is the Belavkin--Staszewski relative entropy, equivalently the maximal $f$-divergence generated by $x\log x$ \cite{belavkinStaszewski1982}. For completeness, the second equality in \eqref{eq:BS_endpoint_new} follows by taking $B=\sigma^{-1/2}\rho^{1/2}$. Then $K_1=BB^*$, and the identity
\begin{equation*}
        B^*\log(BB^*)=\log(B^*B)B^*
\end{equation*}
together with cyclicity of the trace gives
\begin{equation*}
 \tr[\sigma BB^*\log(BB^*)]
 =\tr[\log(B^*B)B^*\sigma B]
 =\tr[\rho\log(B^*B)].
\end{equation*}
When $\rho$ and $\sigma$ commute, both endpoints reduce to the classical relative entropy \cite{hiaiMosonyi2017}. Their Hessians, however, give two different canonical quantum extensions of the Fisher metric: the Bogoliubov--Kubo--Mori metric at $t=0$ \cite{kubo1957,mori1965,petzToth1993,petz1994canonical} and the right logarithmic derivative metric at $t=1$ \cite{yuenLax1973,holevo2011,petz1996monotone}. 

For this generator,
\begin{equation*}
      \widetilde f(x)=x\log x-x+1,\qquad \kappa_f=f''(1)=1.
\end{equation*}
Consequently, Theorem \ref{theo:Hessian_metric_general} gives, for $0\leq t<1$,
\begin{equation}\label{eq:mt_xlogx_first}
 m_t(r)=\frac{(r^{1-t}-1)^2}
 {\widetilde f(r^{1-t})+r^{1-2t}\widetilde f(r^{t-1})},\qquad m_t(1)=1.
\end{equation}
The denominator can be  simplified explicitly. Indeed,
\begin{align*}
 &\widetilde f(r^{1-t})+r^{1-2t}\widetilde f(r^{t-1})\\
 &\quad=r^{1-t}(1-t)\log r-r^{1-t}+1
   +r^{1-2t}\left[-r^{t-1}(1-t)\log r-r^{t-1}+1\right]\\
 &\quad=r^{-t}\left[(1-t)(r-1)\log r
        -(r^t-1)(r^{1-t}-1)\right].
\end{align*}
Thus an equivalent expression is
\begin{equation}\label{eq:mt_xlogx_simplified}
 m_t(r)=\frac{r^t(r^{1-t}-1)^2}
 {(1-t)(r-1)\log r-(r^t-1)(r^{1-t}-1)}.
\end{equation}
All the expressions are understood by continuity at their removable singularities. In particular,
\begin{equation}\label{eq:BKM_representing_function_new}
        m_0(r)=\frac{r-1}{\log r},
\end{equation}
which corresponds to the mean representing function of the BKM metric. Therefore
\begin{equation*}
 g_\rho^{\mathrm{BKM}}(X,Y)
 =\sum_{i,j}\frac{\log p_i-\log p_j}{p_i-p_j}\overline{X_{ij}}Y_{ij},
\end{equation*}
where the quotient is  defined to be $p_i^{-1}$ when $p_i=p_j$.

While for $t=0$ we obtain different metrics by applying different non-affine operator convex functions,  at the other endpoint, the resulting mean depends only on the value $f''(1)$. Put $a=1-t$. Since
\begin{align*}
 r^a-1&=a\log r+O(a^2),\\
 \widetilde f(r^a)&=\frac{f''(1)}{2}a^2(\log r)^2+o(a^2),\\
 r^{2a-1}\widetilde f(r^{-a})
 &=\frac{f''(1)}{2r}a^2(\log r)^2+o(a^2),
\end{align*}
we obtain
\begin{align}
       m_1(r):=\lim_{t\uparrow1}m_t(r)
       &=\kappa_f\lim_{a\downarrow0}
       \frac{a^2(\log r)^2+o(a^2)}
       {\frac{f''(1)}{2}(1+r^{-1})a^2(\log r)^2+o(a^2)}\notag\\
       &=\frac{2r}{1+r}.\label{eq:RLD_representing_function_new}
\end{align}
This is the harmonic-mean representing function of the RLD metric, and hence
\begin{equation*}
 g_\rho^{\mathrm{RLD}}(X,Y)
 =\frac12\sum_{i,j}\left(\frac1{p_i}+\frac1{p_j}\right)
     \overline{X_{ij}}Y_{ij}.
\end{equation*}
This is summarized in the following result, which was shown in \cite{matsumoto2018maximal}, and here we include with a different approach.
\begin{corollary}\label{coro:Coro24}
 Let $f:(0,\infty)\to \mathbb{R}$ be a non-affine operator convex function and $\rho \in \cS(\cH)$ a invertible quantum state.  Then, for $X,Y\in T_\rho\cS(\cH)$,
    \begin{equation}\label{eq:Formula_metric_1}
         g_{\rho,1}^{(f)}(X,Y)=\kappa_f g_{\rho}^{RLD}(X,Y)\, .
    \end{equation} 
\end{corollary}
\begin{proof}
We first identify the Hessian at $t=1$ with the limit of the Hessians for $t<1$. Fix a self-adjoint tangent vector $X$ and, for sufficiently small $\theta$, set
\begin{equation*}
 \rho_s:=\rho+s X.
\end{equation*}
Invertibility of $\rho$ ensures that $\rho_{s}$ remains invertible near $s=0$. In finite dimension, matrix powers and inverse square roots depend smoothly on positive-definite matrices. Hence
\begin{equation*}
 (t,s)\longmapsto
 \Gamma_t(\rho,\rho_s)
 =L_{\rho^{1-t}}
 R_{\rho_s^{-1/2}\rho^t\rho_s^{-1/2}}
\end{equation*}
is smooth in a neighborhood of $(1,0)$. Since an operator-convex function is smooth on $(0,\infty)$ and the spectra involved stay in a compact subset of that interval, finite-dimensional functional calculus shows that
\begin{equation*}
 F_X(t,s):=D_f^t(\rho\Vert\rho_s)
\end{equation*}
is $C^2$ in $s$, with $\partial_s^2F_X(t,0)$ continuous in $t$ at $t=1$.

By the second-variation identity \eqref{eq:Second_derivative_divergence},
\begin{equation*}
 g_{\rho,t}^{(f)}(X,X)
 =\left.\frac{\partial^2}{\partial^2 s}
 F_X(t,s)\right|_{s=0}.
\end{equation*}
Consequently,
\begin{equation}\label{eq:Hessian-continuity-at-one}
 g_{\rho,1}^{(f)}(X,X)
 =\lim_{t\uparrow1}g_{\rho,t}^{(f)}(X,X).
\end{equation}
For $t<1$, Theorem~\ref{theo:Hessian_metric_general} gives
\begin{equation*}
 g_{\rho,t}^{(f)}(X,X)
 =\sum_{i,j}
 \frac{\kappa_f}{p_jm_{f,t}(p_i/p_j)}|X_{ij}|^2.
\end{equation*}
The sum is finite, so \eqref{eq:RLD_representing_function_new} and \eqref{eq:Hessian-continuity-at-one} yield
\begin{align*}
 g_{\rho,1}^{(f)}(X,X)
 &=\sum_{i,j}
 \frac{\kappa_f}{p_jm_{f,1}(p_i/p_j)}|X_{ij}|^2\\
 &=\frac{\kappa_f}{2}\sum_{i,j}
 \left(\frac1{p_i}+\frac1{p_j}\right)|X_{ij}|^2\\
 &=\kappa_f\,g_\rho^{\mathrm{RLD}}(X,X),
\end{align*}
where we used
\begin{equation*}
 p_jm_{f,1}(p_i/p_j)
 =p_j\frac{2(p_i/p_j)}{1+p_i/p_j}
 =\frac{2p_ip_j}{p_i+p_j}.
\end{equation*}
Finally, both sides of \eqref{eq:Formula_metric_1} are symmetric real bilinear forms on the self-adjoint tangent space. Applying the polarization identity to the equality of their quadratic forms,
\begin{equation*}
 g_{\rho,1}^{(f)}(X,Y)=\frac14\bigl(g_{\rho,1}^{(f)}(X+Y,X+Y)-g_{\rho,1}^{(f)}(X-Y,X-Y)\bigr),
\end{equation*}
proves \eqref{eq:Formula_metric_1} for arbitrary $X,Y$.
\end{proof}

This result demonstrates the usefulness of this interpolation: by constructing a continuous path of interpolated relative modular operators that converges to the commutant Radon-Nikodym derivative, we can take limits and study the convergence ratio at $t=1$.

The next theorem shows that the whole family is ordered between these two endpoints and, more strongly, that every quadratic form varies convexly with the interpolation parameter. The proof follows some of the strategies derived in \cite{besenyei2012hasegawa}.

\begin{theorem}[Monotone and convex interpolation from BKM to RLD]
\label{theo:metric-ordering-updated}
Let $0\leq s\leq t\leq1$. Then, for  every invertible state $\rho$ and every self-adjoint traceless $X$,
\begin{equation}\label{eq:metric-ordering-form-updated}
    g_{\rho,s}(X,X)
    \leq g_{\rho,t}(X,X).
\end{equation}
Moreover, for fixed $\rho$ and $X$, the function $t\mapsto g_{\rho,t}(X,X)$ is convex. Consequently,
\begin{equation}\label{eq:BKM-RLD-order-updated}
    g_\rho^{\mathrm{BKM}}
    \leq g_{\rho,t}
    \leq(1-t)g_\rho^{\mathrm{BKM}}
       +t g_\rho^{\mathrm{RLD}}
    \leq g_\rho^{\mathrm{RLD}}.
\end{equation}
The first inequality is strict on every contribution corresponding to a matrix element $X_{ij}\neq0$ with $p_i\neq p_j$ and $t>0$. In particular, the inequality between the quadratic forms is strict whenever at least one such matrix element is  nonzero. On commuting directions, all metrics coincide.
\end{theorem}

\begin{proof}
Let
\begin{equation*}
       \rho=\sum_i p_i|i\rangle\langle i|
\end{equation*}
be a spectral decomposition of $\rho$. By Theorem \ref{theo:Hessian_metric_general},
\begin{equation}\label{eq:metric_xlogx_spectral}
 g_{\rho,t}(X,X)
 =\sum_{i,j}\frac{|X_{ij}|^2}{p_jm_t(p_i/p_j)}.
\end{equation}
It is therefore enough  to study, for fixed $r>0$, the dependence on $t$ of $m_t(r)^{-1}$.

The symmetry $m_t(r)=r m_t(r^{-1})$  implies that the quantity
\begin{equation*}
       \frac{\sqrt r}{m_t(r)}
\end{equation*}
is invariant under $r\mapsto r^{-1}$. We may consequently write $r=e^{2u}$ with $u\geq0$.  We next show that
\begin{equation}\label{eq:Phi_metric_identification}
 \frac{e^u}{m_t(e^{2u})}= \frac{(1-t)u\sinh u-\sinh((1-t)u)\sinh(tu)}
       {\sinh^2((1-t)u)}=:\Phi_t(u),
\end{equation}
where $\Phi_t$ is the function introduced in Lemma \ref{lemma:Phi_monotonicity_convexity}.

Put $a=1-t$. Using $\widetilde f(x)=x\log x-x+1$ in \eqref{eq:mt_xlogx_first}, we have
\begin{align*}
 &\widetilde f(e^{2au})
       +e^{2(2a-1)u}\widetilde f(e^{-2au})\\
 &\quad=2au e^{2au}-e^{2au}+1
       +e^{2(2a-1)u}\left(-2au e^{-2au}-e^{-2au}+1\right)\\
 &\quad=2au\left(e^{2au}-e^{2(a-1)u}\right)
       -\left(e^{2au}+e^{2(a-1)u}\right)
       +1+e^{2(2a-1)u}\\
 &\quad=2e^{(2a-1)u}
       \left[
       2au\sinh u-\cosh u+\cosh((2a-1)u)
       \right]\\
 &\quad=4e^{(2a-1)u}
       \left[
       au\sinh u-\sinh(au)\sinh((1-a)u)
       \right],
\end{align*}
A direct expansion  of the expression in braces gives
\begin{align*}
 &\widetilde f(e^{2au})
       +e^{2(2a-1)u}\widetilde f(e^{-2au})\\
 &\quad=2au e^{2au}-e^{2au}+1
       +e^{2(2a-1)u}\left(-2au e^{-2au}-e^{-2au}+1\right)\\
 &\quad=4e^{(2a-1)u}
        \left[a u\sinh  u-\sinh(au)\sinh((1-a)u)\right].
\end{align*}
Since
\begin{equation*}
        (e^{2au}-1)^2=4e^{2au}\sinh^2(au),
\end{equation*}
the exponential factors cancel and we obtain
\begin{equation}\label{eq:Phi_first_formula}
        \frac{e^u}{m_t(e^{2u})}
        =
       \frac{(1-t)u\sinh u-\sinh((1-t)u)\sinh(tu)}
       {\sinh^2((1-t)u)}
       =\Phi_t(u).
\end{equation}
This proves \eqref{eq:Phi_metric_identification}.  For $r=p_i/p_j$ and
\begin{equation*}
 u_{ij}:=\frac12|\log(p_i/p_j)|,
\end{equation*}
we  have
\begin{equation}\label{eq:coefficient_Phi}
       \frac1{p_jm_t(p_i/p_j)}
        =\frac{\Phi_t(u_{ij})}{\sqrt{p_ip_j}}.
\end{equation}
By Lemma \ref{lemma:Phi_monotonicity_convexity}, for every $u\geq0$ the map $t\mapsto\Phi_t(u)$ is nondecreasing and convex on $[0,1]$, and it is strictly increasing whenever $u>0$.

Substituting \eqref{eq:coefficient_Phi} into \eqref{eq:metric_xlogx_spectral} gives
\begin{equation}\label{eq:metric_as_Phi_sum}
       g_{\rho,t}(X,X)
       =\sum_{i,j}\frac{\Phi_t(u_{ij})}{\sqrt{p_ip_j}}|X_{ij}|^2.
\end{equation}
Every weight in this finite sum  is  nonnegative. The monotonicity and convexity of $t\mapsto g_{\rho,t}(X,X)$ now follow term by term from the corresponding properties of $\Phi_t$. This proves \eqref{eq:metric-ordering-form-updated} and the convexity assertion.

By \eqref{eq:BKM_representing_function_new} and \eqref{eq:RLD_representing_function_new}, the endpoint quadratic forms are
\begin{equation*}
       g_{\rho,0}=g_\rho^{\mathrm{BKM}},\qquad
       g_{\rho,1}=g_\rho^{\mathrm{RLD}}.
\end{equation*}
Monotonicity  gives
\begin{equation*}
       g_\rho^{\mathrm{BKM}}\leq g_{\rho,t}
         \leq g_\rho^{\mathrm{RLD}},
\end{equation*}
while convexity on $[0,1]$ gives the chord bound
\begin{equation*}
       g_{\rho,t}
       \leq(1-t)g_\rho^{\mathrm{BKM}}
       +t g_\rho^{\mathrm{RLD}}.
\end{equation*}
Finally, since $g_\rho^{\mathrm{BKM}}\leq  g_\rho^{\mathrm{RLD}}$,
\begin{equation*}
        (1-t)g_\rho^{\mathrm{BKM}}
       +t g_\rho^{\mathrm{RLD}}
       \leq g_\rho^{\mathrm{RLD}}.
\end{equation*}
This proves \eqref{eq:BKM-RLD-order-updated}.

If $p_i\neq p_j$, then $u_{ij}>0$, and \eqref{eq:Phi_derivative_t} shows that $\Phi_t(u_{ij})>\Phi_0(u_{ij})$ for every $t>0$. Hence the contribution  of  every nonzero matrix element $X_{ij}$ joining two different eigenspaces is strictly larger than its BKM contribution. Conversely, if $[X,\rho]=0$, then
\begin{equation*}
       (p_i-p_j)X_{ij}=0
\end{equation*}
for all $i,j$. Thus $X_{ij}$ can be nonzero only when $p_i=p_j$, in which case $u_{ij}=0$ and $\Phi_t(u_{ij})=1$ independently of $t$. Formula \eqref{eq:metric_as_Phi_sum} then shows that all the metrics coincide on commuting directions.
\end{proof}

\subsection{Connection between the geodesic and Hessian geometries}

The parameter $t$ first enters the construction at the level of the relative operators: it specifies the position along the geodesic joining $\Delta_{\rho\mid\sigma}$ and $\widehat\Delta_{\rho\mid\sigma}$. After taking the Hessian, the resulting metric is instead encoded by the representing function $m_t$. A priori, it is not clear whether the parameter inherited from the geodesic has an intrinsic meaning in the metric geometry, or whether it is only a label for the family. The behavior of $m_t$ at the boundary provides such an intrinsic interpretation. Indeed, by the spectral formula for the metric, the coefficient associated with two eigenvalues $p_i,p_j$ is
\begin{equation*}
       \frac{1}{p_jm_t(p_i/p_j)}.
\end{equation*}
Consequently, the behavior of $m_t(r)$ as $r\downarrow0$ measures the singularity of the metric when two eigenvalues become strongly separated. This suggests considering, whenever the limit exists, the boundary index
\begin{equation*}
       \beta(m):=\lim_{r\downarrow0}\frac{\log m(r)}{\log r}.
\end{equation*}
The next result shows that the geodesic parameter can be recovered exactly from this boundary invariant of the Hessian metric.

\begin{theorem}[The geodesic parameter as boundary index]
\label{theo:boundary-index-updated}
For every $t\in[0,1]$,
\begin{equation}\label{eq:boundary-index-updated}
    \lim_{r\downarrow0}
    \frac{\log m_t(r)}{\log r}=t.
\end{equation}
More precisely, as $r\downarrow0$,
\begin{equation}\label{eq:boundary-asymptotic}
 m_t(r)=\frac{r^t}{(1-t)|\log r|}\bigl(1+o(1)\bigr),\qquad 0\leq t<1,
 \qquad
 m_1(r)=2r\bigl(1+o(1)\bigr).
\end{equation}
\end{theorem}

\begin{proof}
Fix first $0\leq t<1$ and put $L:=-\log r$. Then $L\to\infty$ as $r\downarrow0$. From \eqref{eq:mt_xlogx_simplified},
\begin{equation*}
 m_t(r)=\frac{r^t(r^{1-t}-1)^2}
 {(1-t)(r-1)\log r-(r^t-1)(r^{1-t}-1)}.
\end{equation*}
Since $t<1$, one has $r^{1-t}\to0$, and therefore
\begin{equation*}
       (r^{1-t}-1)^2=1+o(1).
\end{equation*}
For the denominator, observe that
\begin{equation*}
 \frac{(1-t)(r-1)\log r-(r^t-1)(r^{1-t}-1)}{(1-t)L}=
   (1-r)-\frac{(r^t-1)(r^{1-t}-1)}{(1-t)L}.
\end{equation*}
On the one hand, the first term converges to one. On the other hand, the numerator in the second term remains bounded as $r\downarrow0$, whereas $L\to\infty$, and hence the second term converges to zero. It follows that
\begin{equation*}
 (1-t)(r-1)\log r-(r^t-1)(r^{1-t}-1)
       =(1-t)L\bigl(1+o(1)\bigr).
\end{equation*}
Combining the last two estimates proves the first asymptotic relation in \eqref{eq:boundary-asymptotic}.

At $t=1$, equation \eqref{eq:RLD_representing_function_new} gives
\begin{equation*}
       m_1(r)=\frac{2r}{1+r}=2r\bigl(1+o(1)\bigr),
\end{equation*}
which proves the second relation in \eqref{eq:boundary-asymptotic}.

For $0\leq t<1$, taking logarithms in the first relation of \eqref{eq:boundary-asymptotic} gives
\begin{equation*}
       \log m_t(r)=t\log r-\log\bigl((1-t)L\bigr)+o(1).
\end{equation*}
Finally, since $\log r=-L$ and $\log L/L\to0$, we obtain
\begin{equation*}
\begin{aligned}
 \frac{\log m_t(r)}{\log r}
 &=t+\frac{\log((1-t)L)}{L}+o(L^{-1})
 \longrightarrow t.
\end{aligned}
\end{equation*}
For $t=1$,
\begin{equation*}
       \frac{\log m_1(r)}{\log r}
       =1+\frac{\log 2-\log(1+r)}{\log r}
       \longrightarrow1.
\end{equation*}
This proves \eqref{eq:boundary-index-updated}.
\end{proof}

Theorem \ref{theo:boundary-index-updated} identifies the same parameter in the two geometries. On the divergence side, $t$ is the geodesic interpolating $\Delta_{\rho\mid\sigma}$ and $\widehat\Delta_{\rho\mid\sigma}$. On the Hessian side, it is the boundary exponent of the representing function. In particular, the label $t$ can be recovered from the metric alone, without referring back to the divergence from which the metric was obtained.

The more precise asymptotics in \eqref{eq:boundary-asymptotic} also show how the metric becomes singular near the boundary. If $r=p_i/p_j\downarrow0$, then, for $0\leq t<1$,
\begin{equation*}
 \frac{1}{p_jm_t(p_i/p_j)}
 =(1-t)\frac{|\log(p_i/p_j)|}{p_i^t p_j^{1-t}}\bigl(1+o(1)\bigr),
\end{equation*}
whereas at the RLD endpoint
\begin{equation*}
 \frac{1}{p_jm_1(p_i/p_j)}
 =\frac{1+p_i/p_j}{2p_i}
 =\frac{1}{2p_i}\bigl(1+o(1)\bigr).
\end{equation*}
Thus the BKM metric, corresponding to $t=0$, has only a logarithmic boundary divergence, the RLD metric has the reciprocal divergence of order $p_i^{-1}$, and every intermediate metric has the fractional power $p_i^{-t}$, up to the logarithmic correction. This provides a direct connection between the affine-invariant geodesic interpolation and the boundary geometry of its Hessian metrics.

{
\subsection{Comparison with the Hasegawa--Petz metric }
\label{subsec:comparison-known-metric-families}

The operator monotonicity obtained from data processing places the functions $m_t$ inside the Morozova--Chentsov--Petz and Kubo--Ando theories \cite{kuboAndo1980,morozovaChentsov1991}.  For positive definite $A$ and $B$, let
\begin{equation}\label{eq:Kubo-Ando-mean-mt}
 A\mathbin{\mathfrak M_t}B
 :=A^{1/2}m_t(A^{-1/2}BA^{-1/2})A^{1/2}
\end{equation}
be the Kubo--Ando mean represented by $m_t$ \cite{kuboAndo1980}.  At the endpoints,
\begin{equation*}
 A\mathbin{\mathfrak M_0}B=L(A,B),
 \qquad
 A\mathbin{\mathfrak M_1}B=A!B,
\end{equation*}
where $L$ is the logarithmic operator mean and
\begin{equation*}
 A!B:=2(A^{-1}+B^{-1})^{-1}
\end{equation*}
is the harmonic operator mean.

Apart from the interpolation studied in Section \ref{sec:Interpolation_BKM_RLD}, there is, however, another well-known path between the same two endpoints, namely the negative-parameter branch of the Hasegawa--Petz extension of the Wigner--Yanase--Dyson metrics \cite{hasegawa1993,petzHasegawa1996,gibiliscoIsola2004paired}.  The definition is given explicitly in \cite{hasegawa1997non}:
\begin{equation}
    m_p^{\mathrm{HP}}(r)=p(1-p)\frac{(r-1)^2}{(r^p-1)(r^{1-p}-1)}\, , \quad p\in  [-1,2]\, ,
\end{equation}
  where in this interval  the parameter $p$  satisfies monotonicity  \cite[Proposition 5]{besenyei2012hasegawa}. In the parameterization used here,  we let $t=-p$, i.e. the negative branch, so we can then write  it as 
\begin{equation}\label{eq:HP-representing-function}
 m_t^{\mathrm{HP}}(r)
 :=t(1+t)\frac{r^t(r-1)^2}
 {(r^t-1)(r^{1+t}-1)},
\end{equation}
and we  will restrict to the values $t\in [0,1]$. The  continuous endpoint values are
\begin{equation*}
 m_0^{\mathrm{HP}}(r)=\frac{r-1}{\log r},
 \qquad
 m_1^{\mathrm{HP}}(r)=\frac{2r}{1+r}\, ,
\end{equation*}
which coincide with \eqref{eq:BKM_representing_function_new} and \eqref{eq:RLD_representing_function_new}. 
The parameter dependence and mean-theoretic properties of these functions were studied
further by Besenyei \cite{besenyei2012hasegawa}.

This path can be recovered directly from the general Hessian formula of Theorem \ref{theo:Hessian_metric_general}.  For $0<t\leq1$, define the function
\begin{equation}\label{eq:HP-generator-Ft}
 F_t(x)
 :=\frac{x^{1+t}-(1+t)x+t}{t(1+t)},
\end{equation}
and define its continuous value at $t=0$ by
\begin{equation*}
 F_0(x):=x\log x-x+1.
\end{equation*}
Since $x\mapsto x^{1+t}$ is operator convex for $1+t\in[1,2]$, every $F_t$ is operator convex, and
\begin{equation*}
 F_t(1)=F_t'(1)=0,
 \qquad
 F_t''(1)=1.
\end{equation*}
Notice also that
\begin{equation*}
 F_1(x)=\frac12(x-1)^2.
\end{equation*}
To distinguish the two independent deformations, for $0\leq\tau,t\leq1$ we use the notation
\begin{equation}\label{eq:two-parameter-Hessian-family}
 g_{\rho;\tau,t}(X,Y)
 :=g_{\rho,\tau}^{(F_t)}(X,Y).
\end{equation}
The first parameter $\tau$ chooses the relative operator $\Gamma_\tau=\Gamma_0\#_\tau \Gamma_1$, whereas the second parameter $t$ chooses the scalar generator $F_t$.

\begin{proposition}[Two distinct BKM--RLD interpolations]
\label{prop:comparison-Hasegawa-Petz}
Let $0\leq\tau,t\leq1$, and let $g_{\rho;\tau,t}$ be defined by \eqref{eq:two-parameter-Hessian-family}, with the endpoint values understood by continuity.  Then the following statements hold.

\begin{enumerate}[label=\textup{(\roman*)}]
\item Keeping the standard relative modular operator fixed and varying the generator recovers the Hasegawa--Petz path:
\begin{equation}\label{eq:HP-as-standard-edge}
 g_{\rho;0,t}=g_{\rho,t}^{\mathrm{HP}},
 \qquad
 m_{F_t,0}(r)=m_t^{\mathrm{HP}}(r).
\end{equation}

\item Keeping the $F_0$ fixed and varying the relative operator recovers the geodesic path studied in this paper:
\begin{equation}\label{eq:geodesic-as-log-edge}
 g_{\rho;\tau,0}=g_{\rho,\tau}.
\end{equation}

\item The quadratic-generator edge and the maximal-relative-operator edge are both constant and equal to the RLD metric:
\begin{equation}\label{eq:constant-RLD-edges}
 g_{\rho;\tau,1}=g_\rho^{\mathrm{RLD}}
 =g_{\rho;1,t}.
\end{equation}

\item 
At the same parameter value $0<t<1$, the geodesic and Hasegawa--Petz paths are different.  More precisely,
\begin{equation}\label{eq:comparison-boundary-HP}
 \beta(m_t)=\beta(m_t^{\mathrm{HP}})=t,
 \qquad
 \lim_{r\downarrow0}\frac{m_t(r)}{m_t^{\mathrm{HP}}(r)}=0.
\end{equation}
\end{enumerate}
\end{proposition}

\begin{proof}
\textup{(i)}. For $t>0$, the centered function in \eqref{eq:HP-generator-Ft} has $\kappa_{F_t}=F_t''(1)=1$.  At the standard endpoint $\tau=0$, Theorem \ref{theo:Hessian_metric_general} or \cite[eq. 14]{hiai-2012} gives
\begin{equation*}
 m_{F_t,0}(r)
 =\frac{(r-1)^2}{F_t(r)+rF_t(r^{-1})}.
\end{equation*}
A direct computation yields
\begin{align*}
 F_t(r)+rF_t(r^{-1})
 &=\frac{r^{1+t}+r^{-t}-r-1}{t(1+t)}\\
 &=\frac{(r^t-1)(r^{1+t}-1)}{t(1+t)r^t}.
\end{align*}
Consequently,
\begin{equation*}
 m_{F_t,0}(r)
 =t(1+t)\frac{r^t(r-1)^2}
 {(r^t-1)(r^{1+t}-1)}
 =m_t^{\mathrm{HP}}(r),
\end{equation*}
which proves \eqref{eq:HP-as-standard-edge} for $t>0$.  The value at $t=0$ follows from
\begin{equation*}
 \lim_{t\downarrow0}F_t(x)=x\log x-x+1
\end{equation*}
and from the continuous extension of \eqref{eq:HP-representing-function}.  At the level of divergences, this calculation says that the Hasegawa--Petz path is obtained from the standard relative modular operator through
\begin{equation*}
 D_{F_t}^{0}(\rho\Vert\sigma)
 =\frac{\tr\rho^{1+t}\sigma^{-t}-1}{t(1+t)},
 \qquad t>0,
\end{equation*}
where the affine terms in $F_t$ reduce to $-1$ because $\tr\rho=\tr\sigma=1$.

\textup{(ii)}. For $t=0$, the function $F_0(x)=x\log x-x+1$ differs from $x\log x$ only by an affine function.  Affine terms do not change the Hessian, and therefore
\begin{equation*}
 g_{\rho;\tau,0}
 =g_{\rho,\tau}^{(F_0)}
 =g_{\rho,\tau}^{(x\log x)}
 =g_{\rho,\tau},
\end{equation*}
which proves \eqref{eq:geodesic-as-log-edge}.

\textup{(iii)}. We next prove \eqref{eq:constant-RLD-edges}.  Since
\begin{equation*}
 F_1(x)=\frac12(x-1)^2,
\end{equation*}
putting $q=r^{1-\tau}$ in \eqref{eq:Definition_mft} gives
\begin{align*}
 F_1(q)+r^{1-2\tau}F_1(q^{-1})
 &=\frac12(q-1)^2
   +\frac12r^{1-2\tau}\frac{(q-1)^2}{q^2}\\
 &=\frac12(q-1)^2\left(1+\frac1r\right).
\end{align*}
Hence
\begin{equation*}
 m_{F_1,\tau}(r)=\frac{2r}{1+r}
\end{equation*}
for every $0\leq\tau<1$. For $\tau=1$, the result follows by Corollary \ref{coro:Coro24}


\textup{(iv)}.  For fixed $0<t<1$, equation \eqref{eq:HP-representing-function} gives
\begin{equation*}
 m_t^{\mathrm{HP}}(r)
 =t(1+t)r^t\bigl(1+o(1)\bigr),
 \qquad r\downarrow0.
\end{equation*}
Combining this with \eqref{eq:boundary-asymptotic} yields
\begin{equation*}
 \frac{m_t(r)}{m_t^{\mathrm{HP}}(r)}
 =\frac{1+o(1)}
 {t(1+t)(1-t)|\log r|}
 \longrightarrow0.
\end{equation*}
Both representing functions have boundary exponent $t$, but their leading asymptotics differ by a logarithmic factor.  This proves \eqref{eq:comparison-boundary-HP}.
\end{proof}

Proposition \ref{prop:comparison-Hasegawa-Petz} can be summarized by
the two-parameter deformation square in
Figure \ref{fig:HP-geodesic-square}. The horizontal coordinate
$\tau$ changes the relative operator, while the vertical coordinate
$t$ changes the scalar generator.

\begin{figure}[t]
\centering
\begin{tikzpicture}[
  >=Latex,
  metric/.style={
    draw=black!70,
    rounded corners=2pt,
    fill=white,
    align=center,
    inner sep=6pt,
    text width=3.25cm,
    minimum height=1.0cm,
    font=\small
  },
  present/.style={
    -{Latex[length=2.8mm]},
    very thick,
    blue!70!black
  },
  hp/.style={
    -{Latex[length=2.8mm]},
    very thick,
    orange!85!black
  },
  constant/.style={
    -{Latex[length=2.5mm]},
    thick,
    dashed,
    black!55
  },
  lab/.style={
    fill=white,
    inner sep=2pt,
    align=center,
    font=\footnotesize
  }
]

\node[metric] (bkm) at (0,0)
 {$g_{\rho;0,0}=g_\rho^{\mathrm{BKM}}$\\[-1mm]
  };

\node[metric] (rldg) at (7.7,0)
 {$g_{\rho;1,0}=g_\rho^{\mathrm{RLD}}$\\[-1mm]
  };

\node[metric] (rldhp) at (0,4.6)
 {$g_{\rho;0,1}=g_\rho^{\mathrm{RLD}}$\\[-1mm]
  };

\node[metric] (rldboth) at (7.7,4.6)
 {$g_{\rho;1,1}=g_\rho^{\mathrm{RLD}}$\\[-1mm]
  };

\draw[present]
 (bkm.east) --
 node[lab,below=4pt]
 {\textbf{geodesic path}\\[-1mm]
  $\tau:0\to1$; vary $\Gamma_\tau$, keep $F_0$ fixed}
 (rldg.west);

\draw[hp]
 (bkm.north) --
 node[lab,left=4pt]
 {\textbf{Hasegawa--Petz path}\\[-1mm]
  $t:0\to1$; vary $F_t$, keep $\Gamma_0$ fixed}
 (rldhp.south);

\draw[constant]
 (rldhp.east) --
 node[lab,above=4pt]
 {$F_1(x)=\tfrac12(x-1)^2$\\[-1mm]
  RLD for every $\tau$}
 (rldboth.west);

\draw[constant]
 (rldg.north) --
 node[lab,right=4pt]
 {$\Gamma_1=\widehat\Delta_{\rho\mid\sigma}$\\[-1mm]
  RLD for every $t$}
 (rldboth.south);

\node[
  draw=black!35,
  rounded corners=2pt,
  fill=black!2,
  align=center,
  inner sep=5pt,
  font=\small
]
 at (3.85,2.3)
 {$g_{\rho;\tau,t}:=g_{\rho,\tau}^{(F_t)}$\\[-1mm]
  {\scriptsize $\tau$: relative-operator deformation;
   $t$: generator deformation}};

\end{tikzpicture}

\caption{
The two distinct BKM--RLD interpolations inside the family
$g_{\rho;\tau,t}=g_{\rho,\tau}^{(F_t)}$.
The lower blue edge is the relative-operator geodesic of the present
paper, with the logarithmic generator $F_0$ fixed. The left orange
edge is the Hasegawa--Petz path, with the standard relative operator
$\Gamma_0$ fixed. The upper and right dashed edges are constant and
equal to the RLD metric.
}
\label{fig:HP-geodesic-square}
\end{figure}
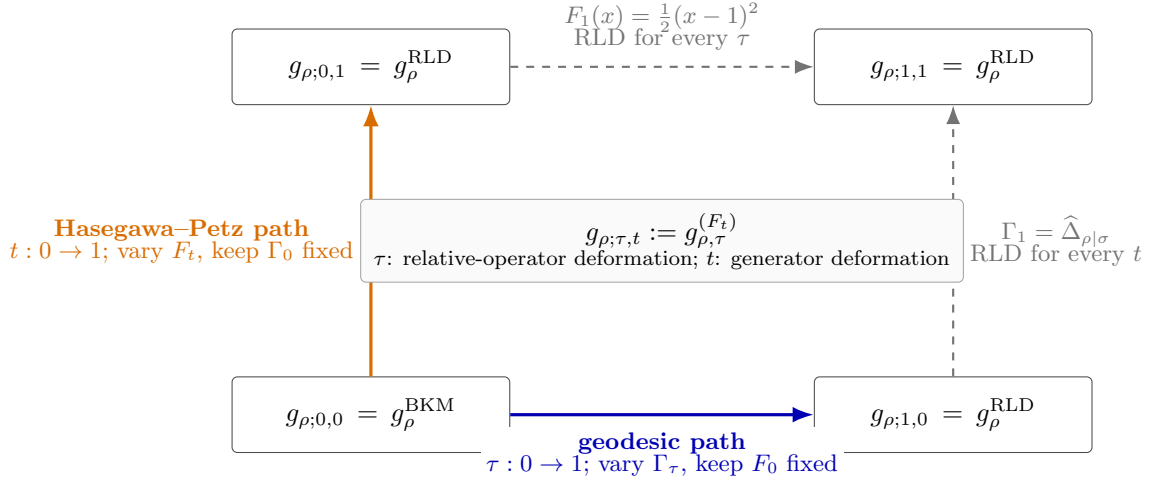

The lower horizontal edge is the geodesic interpolation of this paper: the function  $F_0$ is fixed, while $\Gamma_\tau=\Gamma_0\#_\tau \Gamma_1$ moves from $\Gamma
_0=\Delta_{\rho\mid\sigma}$ to $\Gamma
_1=\widehat\Delta_{\rho\mid\sigma}$.  The left vertical edge is the Hasegawa--Petz interpolation: the standard relative modular operator $\Gamma_0$ is fixed, while the score changes from the logarithmic generator $F_0$ to the quadratic generator $F_1$.  The upper edge is constant because the quadratic generator produces RLD for every relative-operator parameter, and the right edge is constant because the maximal relative operator produces RLD for every generator in the family.  Thus the two paths reach the same endpoint by deforming different ingredients of the divergence.

The common boundary exponent does not erase this structural difference.  For $0<t<1$,
\begin{equation*}
 m_t(r)\sim\frac{r^t}{(1-t)|\log r|},
 \qquad
 m_t^{\mathrm{HP}}(r)\sim t(1+t)r^t,
 \qquad r\downarrow0.
\end{equation*}
Hence both metrics have the same coarse power-distribution sensitivity $r^{-t}$ near the boundary, but the geodesic metric contains an additional logarithmic amplification in its kernel $1/m_t(r)$.

\section{The likelihood surface and the  classicalization of the experiments}\label{sec:Section_2}

Let $(\cX,\cF)$ be a measurable space, where $\cX$ is the set of all possible outcomes (the sample space) and $\cF$ is a $\sigma$-algebra of subsets of $\cX$. Let $P_0$ and $P_1$ be two probability measures corresponding to the null hypothesis $H_0$ and the alternative hypothesis $H_1$, respectively. Suppose that $P_0$ and $P_1$ are dominated by a common $\sigma$-finite and $\sigma$-additive measure $\mu$, i.e. if $A \in \mathcal{F}$ and $\mu(A)=0$, then $P_0(A)=P_1(A)=0$. By the Radon--Nikodym Theorem \cite[Theorem 3.8]{folland1999}, they admit densities
\begin{equation*}
f_0(z)=\frac{\dd P_0}{\dd \mu}(z)\, , \quad  f_1(z)=\frac{\dd P_1}{\dd \mu}(z)\, .
\end{equation*}
We will additionally assume that $P_0$ and $P_1$ are mutually absolutely
continuous, written $P_0\sim P_1$, so that the likelihood ratio is strictly
positive and the log-likelihood is finite almost everywhere; exponential
moments will be considered on their effective domain.

  The basic statistical object for distinguishing the two hypotheses is then the likelihood ratio \cite[Section 3.2]{lehmannRomano2005}.
\begin{equation*}
\Lambda(z)=\frac{f_1(z)}{f_0(z)}=\frac{\dd P_1}{\dd P_0}(z)\, ,
\end{equation*}
defined $P_0$-almost everywhere.  We use the convention $\Lambda=0$ on a $P_0$-null set if a pointwise representative is needed.
 Since $P_1$ is a probability measure, the likelihood ratio satisfies the normalization identity
\begin{equation*}
\bE_{P_0}[\Lambda]=\int_{\cX}\Lambda(z)\dd P_0(z)=1\, .
\end{equation*}
The Kullback--Leibler divergence is \cite{kullbackLeibler1951} 
\begin{equation*}
 D_{\mathrm{KL}}(P_1\Vert P_0)
 :=\int_{\cX}\log\Lambda(z)\,\dd P_1(z)
 =\int_{\cX}\Lambda(z)\log\Lambda(z)\,\dd P_0(z)\, .
\end{equation*}
 Rather than treating information quantities such as the Kullback--Leibler (KL) divergence, the Fisher information or the Chernoff information as unrelated quantities, it is possible to derive all these information quantities in terms of the log-likelihood cumulant function 
\begin{equation*}
\psi_{P_1,P_0}(s)=\log \bE_{P_0}[\Lambda^s]=  \log\int_{\cX}f_1(z)^sf_0(z)^{1-s}\dd \mu (z)\, ,
\end{equation*}
which can be also expressed in terms of the classical R\'enyi $s$-divergence $D_s(P_1\Vert P_0)$, $s\in (0,1)\cup (1,\infty)$, $\psi_{P_1,P_0}(s)=(s-1)D_s(P_1\Vert P_0)$ \cite{renyi1961}. In terms of the corresponding tilted distribution $P_s$   defined as 
\begin{equation*}
\frac{\dd P_s}{\dd P_0}(z)=\exp(s\log \Lambda(z)-\psi_{P_1,P_0}(s))=\frac{\Lambda (z)^s}{\bE_{P_0}[\Lambda^s]}\, ,
\end{equation*}
 we can rewrite the log-likelihood cumulant function as
\begin{equation}
    \psi_{P_1,P_0}(s)=\log \bE_{P_0}[e^{s\log \Lambda}]\, .
\end{equation}
In particular, $(P_s)_{s \in [0,1]}$ is the classical exponential family interpolating between $P_{s=0}=P_0$ and $P_{s=1}=P_1$ \cite{brown1986,amariNagaoka2000,demboZeitouni1998}. 

Whenever the corresponding derivatives of $\psi_{P_1,P_0}(s)$ are finite, this function satisfies
\begin{equation*}
\psi_{P_1,P_0}(0)=\psi_{P_1,P_0}(1)=0\, , \quad \psi_{P_1,P_0}'(1)=D_{KL}(P_1\Vert P_0)\, ,
\end{equation*}
and
\begin{equation*}
\cI(s)=\Var_{P_s}\left( \log \frac{\dd P_1}{\dd P_0} \right)= \psi_{P_1,P_0}''(s)\, ,
\end{equation*}
Moreover,
\begin{equation*}
C(P_0,P_1)=-\min_{0\leq s \leq 1} \psi_{P_1,P_0}(s)
\end{equation*}
is the Chernoff information associated with the binary hypothesis-testing problem $(P_0,P_1)$. This quantity determines the optimal asymptotic exponential decay of the Bayesian error probability in binary i.i.d. hypothesis testing, see e.g. \cite{chernoff1952} or \cite[Chapter 11]{coverThomas2006}.

\subsection{The canonical classical experiment}

The classical discussion above is organized around a single likelihood ratio.  For two noncommuting quantum states there is no distinguished operator quotient $\rho/\sigma$.  In the present construction, however, each point $\Gamma_t$ of the relative-operator geodesic in \eqref{eq:relative-operator-geodesic-updated} provides a canonical positive operator which can be interpreted as a likelihood ratio after taking its spectral distribution.  In addition, the left and right multiplication operators defining $\Gamma_t$ commute and are strictly positive and self-adjoint on the Hilbert--Schmidt space, so $\Gamma_t$ is strictly positive and admits the usual  functional calculus.  Under the standard identification $\mathcal B_2(\cH)\simeq\cH\otimes\overline{\cH}$, the unit vector $\Omega_\sigma$ will denote the canonical Hilbert--Schmidt purification of $\sigma$.

As a quantum analogue of the log-likelihood cumulant,   we introduce the two-parameter surface
\begin{equation}\label{eq:intro-master-surface}
    \Psi_{\rho,\sigma}(t,s)
    :=\log\left\langle
       \Omega_\sigma,\Gamma_t^s\Omega_\sigma
      \right\rangle_{\mathrm{HS}},
\end{equation}
for the study of the  geodesic $f$-divergences, which we will call the log-likelihood surface. The role of the two parameters is different.  The parameter $t$ selects a point of the quantum relative-operator geodesic, and hence a classical experiment associated with the spectral measurement of $\Gamma_t$.  Once $t$ is fixed, the parameter $s$ is the usual exponential-tilting parameter of that experiment.  The following theorem makes this statement precise.  It may be viewed as a spectral-measure extension, along the whole geodesic, of the classicalization of Petz quasi-entropies by the relative modular operator \cite{nussbaumSzkola2009}, which was shown for $f$-divergences in \cite{androulakisJohn2024}, and for maximal $f$-divergences and the commutant Radon-Nykodim derivative in \cite{matsumoto2018maximal}.

Let
\begin{equation}\label{eq:Tt-spectral-decomposition}
\Gamma_t=\sum_{z\in\mathcal Z_t}
\lambda_{t,z}E_{t,z}
\end{equation}
be the spectral decomposition of $\Gamma_t$, where the numbers $\lambda_{t,z}>0$ are distinct and the $E_{t,z}$ are the corresponding spectral projections.
Here $\mathcal Z_t$ is a finite index set containing one label for each distinct eigenvalue of $\Gamma_t$. The projections satisfy
\begin{equation*}
 E_{t,z}E_{t,z'}=\delta_{z,z'}E_{t,z},\qquad
 \sum_{z\in\mathcal Z_t}E_{t,z}=I_{\mathcal B_2(\cH)},
\end{equation*}
and $\{E_{t,z}\}_{z\in\mathcal Z_t}$ is the spectral projection-valued measure (PVM) of $\Gamma_t$.

Define
\begin{equation}\label{eq:pt-qt-definitions}
q_t(z)
:=\langle\Omega_\sigma,
E_{t,z}\Omega_\sigma\rangle_{\mathrm{HS}},
\qquad
p_t(z):=\lambda_{t,z}q_t(z).
\end{equation}
Outcomes with $q_t(z)=0$ may be removed from $\mathcal Z_t$, since by definition, $p_t(z)=0$ for such outcomes as well. We will denote
\begin{equation*}
 \mathcal Z_t^+:=\{z\in\mathcal Z_t:q_t(z)>0\}\, ,
\end{equation*}
and  then the likelihood ratio will only be used on $\mathcal Z_t^+$, where both $p_t$ and $q_t$ are strictly positive.  Whenever a sum is written over the full spectral set $\mathcal Z_t$, an element $z \in \mathcal{Z}_t$ such that $q_t(z)=0$  is understood through the spectral expression as: $q_t(z)f(p_t(z)/q_t(z))$ means $q_t(z)f(\lambda_{t,z})=0$, and $p_t(z)^sq_t(z)^{1-s}$ means $\lambda_{t,z}^sq_t(z)=0$.  Equivalently, all such sums may be restricted to $\mathcal Z_t^+$.  This convention makes \eqref{eq:exact-classicalization-updated} and \eqref{eq:Psi-classical-cumulant-updated} unambiguous for all real $s$ while retaining the full spectral decomposition of the identity.

\begin{theorem}[Geodesic likelihood-surface theorem]
\label{theo:geodesic-likelihood-surface}
Let $\rho,\sigma$ be invertible states and let $t\in[0,1]$. Then
$\Omega_{\rho,t}:=\Gamma_t^{1/2}\Omega_\sigma$ is a purification of $\rho$, and the spectral
measurement of $\Gamma_t$ on $\Omega_{\rho,t}$ and $\Omega_\sigma$
produces the probability distributions $p_t$ and $q_t$, respectively.
Moreover,
\begin{equation}\label{eq:likelihood-ratio-spectrum}
    \frac{p_t(z)}{q_t(z)}=\lambda_{t,z}.
\end{equation}
is the likelihood ratio of the experiment. Consequently, for every operator convex function $f:(0,\infty)\to \mathbb{R}$,
\begin{equation}\label{eq:exact-classicalization-updated}
    D_f^t(\rho\Vert\sigma)
    =
    D_f(p_t\Vert q_t),
\end{equation}
where $D_f$ denotes the classical $f$-divergence, and
\begin{equation}\label{eq:Psi-classical-cumulant-updated}
    \Psi_{\rho,\sigma}(t,s)
    =
    \log\sum_{z} p_t(z)^s q_t(z)^{1-s},
    \qquad s\in\mathbb{R}.
\end{equation}
In particular,
\begin{equation}\label{eq:Psi-normalization-boundary}
    \Psi_{\rho,\sigma}(t,0)
    =\Psi_{\rho,\sigma}(t,1)=0.
\end{equation}
Thus, for every $t\in [0,1]$,  $\Gamma_t$ defines a classical binary experiment whose likelihood ratio
is its spectrum and which exactly represents both the geodesic
$f$-divergences and the likelihood-cumulant profile.
\end{theorem}

\begin{proof}
Fix $t\in[0,1]$.  Since the $E_{t,z}$ are orthogonal projections, the quantities $q_t(z)$ in \eqref{eq:pt-qt-definitions} are nonnegative.  Moreover, the spectral projections resolve the identity on $\mathcal B_2(\cH)$, and hence
\begin{equation*}
 \sum_{z\in\mathcal Z_t}q_t(z)
 =\langle\Omega_\sigma,\Omega_\sigma\rangle_{\mathrm{HS}}
 =\tr\sigma=1.
\end{equation*}
Moreover, the quantities $p_t(z)$ are also nonnegative because every $\lambda_{t,z}$ is strictly positive.  To prove their normalization, first observe from \eqref{eq:relative-operator-geodesic-updated} that
\begin{align*}
 \Gamma_t\Omega_\sigma
 &=\rho^{1-t}\sigma^{1/2}
   \bigl(\sigma^{-1/2}\rho^t\sigma^{-1/2}\bigr)\\
 &=\rho^{1-t}\rho^t\sigma^{-1/2}
 =\rho\sigma^{-1/2}.
\end{align*}
Using \eqref{eq:Tt-spectral-decomposition} and the cyclicity of the trace, we obtain
\begin{align*}
 \sum_{z\in\mathcal Z_t}p_t(z)
 &=\sum_{z\in\mathcal Z_t}\lambda_{t,z}
   \langle\Omega_\sigma,E_{t,z}\Omega_\sigma\rangle_{\mathrm{HS}}\\
 &=\langle\Omega_\sigma,\Gamma_t\Omega_\sigma\rangle_{\mathrm{HS}}\\
 &=\tr\bigl(\sigma^{1/2}\rho\sigma^{-1/2}\bigr)\\
 & =1 \, .
\end{align*}
Thus $p_t$ and $q_t$ are probability distributions.  The identity \eqref{eq:likelihood-ratio-spectrum} follows immediately from the definition $p_t(z)=\lambda_{t,z}q_t(z)$ whenever $q_t(z)>0$.

  By \eqref{eq:relative-operator-geodesic-updated}, $L_{\rho^{1-t}}$ and $R_{K_t}$ commute and are positive, so their positive square roots also commute and
\begin{equation*}
 \Gamma_t^{1/2}
 =L_{\rho^{(1-t)/2}}R_{K_t^{1/2}}.
\end{equation*}
Consequently, as a Hilbert--Schmidt vector,
\begin{equation*}
 \Omega_{\rho,t}
 =\rho^{(1-t)/2}\sigma^{1/2}K_t^{1/2}.
\end{equation*}
Under the vectorization identification $\mathcal B_2(\cH)\simeq\cH\otimes\overline{\cH}$, a Hilbert--Schmidt vector $X$ has left reduced density matrix $XX^*$.  Hence the left reduced density matrix of $\Omega_{\rho,t}$ is
\begin{align*}
 &\rho^{(1-t)/2}\sigma^{1/2}K_t^{1/2}
 K_t^{1/2}\sigma^{1/2}\rho^{(1-t)/2}\\
 &\qquad=\rho^{(1-t)/2}\sigma^{1/2}
 \bigl(\sigma^{-1/2}\rho^t\sigma^{-1/2}\bigr)
 \sigma^{1/2}\rho^{(1-t)/2}\\
 &\qquad=\rho\, .
\end{align*}
This proves that $\Omega_{\rho,t}$ is a purification of $\rho$.  The distribution obtained by measuring $\{E_{t,z}\}_z$ on $\Omega_\sigma$ is $q_t$ by definition.  Since every $E_{t,z}$ is a spectral projection of $\Gamma_t$, it commutes with $\Gamma_t^{1/2}$ and $
 \Gamma_t^{1/2}E_{t,z}\Gamma_t^{1/2}
 =\lambda_{t,z}E_{t,z}$. 
It follows that
\begin{align*}
 \langle\Omega_{\rho,t},E_{t,z}\Omega_{\rho,t}\rangle_{\mathrm{HS}}
 &=\langle\Omega_\sigma,
 \Gamma_t^{1/2}E_{t,z}\Gamma_t^{1/2}\Omega_\sigma\rangle_{\mathrm{HS}}\\
 &=\lambda_{t,z}
 \langle\Omega_\sigma,E_{t,z}\Omega_\sigma\rangle_{\mathrm{HS}}\\
 &=p_t(z),
\end{align*}

Now, we compute the divergence. The  functional calculus gives
\begin{equation*}
 f(\Gamma_t)=\sum_{z\in\mathcal Z_t}f(\lambda_{t,z})E_{t,z}.
\end{equation*}
Therefore, using the definition of the geodesic $f$-divergence and \eqref{eq:likelihood-ratio-spectrum},
\begin{align*}
 D_f^t(\rho\Vert\sigma)
 &=\langle\Omega_\sigma,f(\Gamma_t)\Omega_\sigma\rangle_{\mathrm{HS}}\\
 &=\sum_{z\in\mathcal Z_t}
 f(\lambda_{t,z})
 \langle\Omega_\sigma,E_{t,z}\Omega_\sigma\rangle_{\mathrm{HS}}\\
 &=\sum_{z\in\mathcal Z_t}q_t(z)
 f\!\left(\frac{p_t(z)}{q_t(z)}\right),
\end{align*}
which is exactly \eqref{eq:exact-classicalization-updated}.

Finally, for every real $s$, the strict positivity of $\Gamma_t$ and its spectral decomposition give
\begin{align*}
 \langle\Omega_\sigma,\Gamma_t^s\Omega_\sigma\rangle_{\mathrm{HS}}
 &=\sum_{z\in\mathcal Z_t}\lambda_{t,z}^s q_t(z)\\
 &=\sum_{z\in\mathcal Z_t}
 \left(\frac{p_t(z)}{q_t(z)}\right)^s q_t(z)\\
 &=\sum_{z\in\mathcal Z_t}p_t(z)^s q_t(z)^{1-s}.
\end{align*}
Taking logarithms proves \eqref{eq:Psi-classical-cumulant-updated}.  The boundary identities in \eqref{eq:Psi-normalization-boundary} follow from the normalization of $q_t$ and $p_t$, respectively. 
\end{proof}

For every fixed $t$, define
\begin{equation}\label{eq:canonical-likelihood-notation}
 \Lambda_t(z):=\frac{p_t(z)}{q_t(z)}=\lambda_{t,z},
 \qquad \ell_t(z):=\log\Lambda_t(z).
\end{equation}
Then
\begin{equation}\label{eq:Psi_log_expectation_exponential}
 \Psi_{\rho,\sigma}(t,s)
 =\log\bE_{q_t}[\Lambda_t^s]
 =\log\bE_{q_t}[e^{s\ell_t}],
\end{equation}
so the $s$-section of the surface is the ordinary cumulant-generating function of the log-likelihood ratio of the classical experiment $(p_t,q_t)$.  Its exponentially tilted distribution is
\begin{equation}\label{eq:canonical-tilted-distribution}
 \pi_{t,s}(z)
 :=\exp(s\ell_t(z)-\Psi_{\rho,\sigma}(t,s))q_t(z)
 =\frac{p_t(z)^sq_t(z)^{1-s}}
 {\sum_y p_t(y)^sq_t(y)^{1-s}}.
\end{equation}
with $\pi_{t,0}=q_t$ and $\pi_{t,1}=p_t$.  Thus $t$ moves through the family of canonical experiments selected by the quantum geodesic, while $s$ moves along the classical exponential family inside each such experiment.

At the standard endpoint $t=0$, this construction recovers the Nussbaum--Szko\l a distributions \cite{nussbaumSzkola2009}, up to grouping outcomes which have the same likelihood-ratio value.  Indeed, if
\begin{equation*}
 \rho=\sum_i\alpha_i|u_i\rangle\langle u_i|,
 \qquad
 \sigma=\sum_j\beta_j|v_j\rangle\langle v_j|,
\end{equation*}
then the Hilbert--Schmidt vectors $|u_i\rangle\langle v_j|$ are eigenvectors of $\Gamma_0=L_\rho R_{\sigma^{-1}}$ with eigenvalues $\alpha_i/\beta_j$, and the corresponding ungrouped distributions are
\begin{equation*}
 q_0(i,j)=\beta_j|\langle u_i,v_j\rangle|^2,
 \qquad
 p_0(i,j)=\alpha_i|\langle u_i,v_j\rangle|^2.
\end{equation*}
The theorem therefore extends this familiar exact classicalization from the standard relative modular operator to every point of the relative-operator geodesic.

At $t=1$, write
\begin{equation}
 K_1=\sigma^{-1/2}\rho\sigma^{-1/2}
 =\sum_z\lambda_{1,z}F_z.
\end{equation}
Then $\Gamma_1=R_{K_1}$ and its spectral projections on $\mathcal B_2(\mathcal H)$ are $E_{1,z}=R_{F_z}$.  Consequently,
\begin{equation}
 q_1(z)=\operatorname{tr}(\sigma F_z),
 \qquad
 p_1(z)=\lambda_{1,z}q_1(z).
\end{equation}
These are exactly the distributions obtained by Matsumoto, with his $M_x$ identified with $F_z$, see \cite[Section 4.3]{matsumoto2018maximal}.  A recent explicit treatment of the same maximal distributions is given in \cite{lanierBeguinotRioul2025}.  

\subsection{Classical statistical consequences of the likelihood surface}
\label{subsec:classical-statistical-consequences}

{
Theorem \ref{theo:geodesic-likelihood-surface} is the point at which the noncommutative pair $(\rho,\sigma)$ is converted into the finite classical experiment $(p_t,q_t)$.  Once $t$ is fixed, the remaining likelihood calculus belongs to the classical theory of binary experiments.  The likelihood-ratio characterization of sufficient statistics and the comparison of experiments go back to the foundational work in \cite{halmosSavage1949,kullbackLeibler1951,blackwell1953,lehmannRomano2005}, while the divergence viewpoint was developed systematically in \cite{csiszar1967}.  The relevant asymptotic testing and large-deviation results are the classical theorems of Chernoff, Hoeffding, Strassen and their information-theoretic developments \cite{chernoff1952,hoeffding1965,strassen1962,csiszarLongo1971,blahut1974,coverThomas2006,demboZeitouni1998}.  We collect only the consequences needed later in the paper, emphasizing how they are derived by the likelihood surface $\Psi_{\rho,\sigma}(t,s)$.}

As before, we will restrict to the set $\mathcal Z_t^+$.  Let
\begin{equation*}
 \Pi_t:=\sum_{z\in\mathcal Z_t^+}E_{t,z}
\end{equation*}
and regard $\Gamma_t$ as acting on the  Hilbert--Schmidt subspace $\Pi_t\mathcal B_2(\cH)$.  The commutative algebra generated by the likelihood observable is
\begin{equation}\label{eq:likelihood-algebra}
 \cA_t:=W^*(\Pi_t \Gamma_t \Pi_t)
 =\left\{\sum_{z\in\mathcal Z_t^+}a(z)E_{t,z}:a(z)\in\mathbb C\right\}.
\end{equation}
Here $W^*(\Gamma_t)$ denotes the finite-dimensional algebra generated by $\Gamma_t$, i.e. the span of its distinct spectral projections.  Define also
\begin{equation*}
 P_{\rho,t}:=|\Omega_{\rho,t}\rangle\!\langle\Omega_{\rho,t}|,
 \qquad
 P_{\sigma}:=|\Omega_\sigma\rangle\!\langle\Omega_\sigma|.
\end{equation*}

    The following operational statements concern the auxiliary, pair-dependent purification experiment
and tests restricted to the commutative likelihood algebra $\mathcal{A}_t^{\otimes n}$, in other words,  they concern the classical experiment  $(p_t^{\otimes n},q_t^{\otimes n})$. However, they do not describe unrestricted hypothesis testing of the original pair $(\rho^{\otimes n}, \sigma^{\otimes n})$.

\begin{proposition}[Minimal sufficiency and exact finite-blocklength reduction]
\label{prop:exact-classical-test-reduction}
For the binary experiment $(p_t,q_t)$, the likelihood ratio $\Lambda_t=p_t/q_t$ is a minimal sufficient statistic.  Moreover, for every $n\geq1$, test operators in $\cA_t^{\otimes n}$ are in one-to-one correspondence with randomized classical tests for $p_t^{\otimes n}$ against $q_t^{\otimes n}$.  More precisely, every test $0\leq Q_n\leq I$ in $\cA_t^{\otimes n}$ has a unique representation
\begin{equation}\label{eq:classical-test-representation}
 Q_n=\sum_{\mathbf z\in(\mathcal Z_t^+)^n}
 \tau_n(\mathbf z)E_{t,\mathbf z},
 \qquad 0\leq\tau_n(\mathbf z)\leq1,
\end{equation}
where
\begin{equation*}
 E_{t,\mathbf z}:=E_{t,z_1}\otimes\cdots\otimes E_{t,z_n}.
\end{equation*}
For this correspondence,
\begin{align}
 \tr[Q_nP_{\rho,t}^{\otimes n}]
 &=\sum_{\mathbf z}\tau_n(\mathbf z)p_t^{\otimes n}(\mathbf z),
 \label{eq:finite-n-p-reduction}\\
 \tr[Q_nP_{\sigma}^{\otimes n}]
 &=\sum_{\mathbf z}\tau_n(\mathbf z)q_t^{\otimes n}(\mathbf z).
 \label{eq:finite-n-q-reduction}
\end{align}
\end{proposition}

\begin{proof}

For a binary experiment, the sufficiency of $\Lambda_t$ is shown for example in \cite[Theorem 6.12]{lehmann1998theory}. The projections $E_{t,\mathbf z}$ are the mutually orthogonal minimal projections of the finite-dimensional commutative algebra $\cA_t^{\otimes n}$ and resolve its identity.  Therefore every self-adjoint element of this algebra is diagonal in these projections, and the test inequalities $0\leq Q_n\leq I$ are equivalent to $0\leq\tau_n(\mathbf z)\leq1$.  Theorem \ref{theo:geodesic-likelihood-surface}  gives
\begin{equation*}
 \langle\Omega_{\rho,t},E_{t,z}\Omega_{\rho,t}\rangle_{\mathrm{HS}}=p_t(z),
 \qquad
 \langle\Omega_\sigma,E_{t,z}\Omega_\sigma\rangle_{\mathrm{HS}}=q_t(z).
\end{equation*}
Taking tensor products and inserting \eqref{eq:classical-test-representation} proves \eqref{eq:finite-n-p-reduction} and \eqref{eq:finite-n-q-reduction}.
\end{proof}

For hypothesis testing, we use the convention that $Q_n$ accepts the hypothesis $P_{\rho,t}^{\otimes n}$ and put
\begin{align}
 \alpha_n(Q_n)&:=1-\tr[Q_nP_{\rho,t}^{\otimes n}],\notag\\
 \beta_n(Q_n)&:=\tr[Q_nP_{\sigma}^{\otimes n}].
 \label{eq:restricted-errors-updated}
\end{align}
For $0<\varepsilon<1$, define
\begin{equation}\label{eq:restricted-beta-updated}
 \beta_{n,\varepsilon}^{(t)}
 :=\inf\left\{\beta_n(Q_n):
 Q_n\in\cA_t^{\otimes n},\ 0\leq Q_n\leq I,\
 \alpha_n(Q_n)\leq\varepsilon\right\},
\end{equation}
and let
\begin{equation}\label{eq:restricted-Chernoff-error}
 P_{\mathrm e,n}^{(t)}
 :=\inf_{\substack{Q_n\in\cA_t^{\otimes n}\\0\leq Q_n\leq I}}
 \frac12\left(1-\tr[Q_nP_{\rho,t}^{\otimes n}]
 +\tr[Q_nP_{\sigma}^{\otimes n}]\right)
\end{equation}
be the corresponding equal-prior Bayesian error. Next result shows that all the main statistical quantities of the experiment associated to some $t\in [0,1]$ can be expressed in terms of the likelihood surface. For $t=0$, these results can be found in e.g. \cite{jaksicEtAl2012}. The proof is a consequence of the classical results, but it is  added for completeness. 

\begin{proposition}[Statistical calculus of $\Psi$]
\label{prop:classical-likelihood-calculus}
Fix $t\in[0,1]$.  Then the following statements hold.

\begin{enumerate}[label=\textup{(\roman*)}]
\item For every $s\in\mathbb R$,
\begin{equation}\label{eq:Psi-first-second-derivative-classical}
 \partial_s\Psi_{\rho,\sigma}(t,s)
 =\bE_{\pi_{t,s}}[\ell_t]\, , \quad \partial_s^2\Psi_{\rho,\sigma}(t,s)
 =\Var_{\pi_{t,s}}(\ell_t)\, .
\end{equation}
In particular,
\begin{equation}\label{eq:Dt-Vt-from-Psi}
 \partial_s\Psi_{\rho,\sigma}(t,1)=D^t(\rho\Vert\sigma)
 =D_{\mathrm{KL}}(p_t\Vert q_t),
\end{equation}
and
\begin{equation}\label{eq:Vt-definition-updated}
 V_t(\rho\Vert\sigma)
 :=\partial_s^2\Psi_{\rho,\sigma}(t,1)
 =\Var_{p_t}\!\left(\log\frac{p_t}{q_t}\right).
\end{equation}

\item The classical Chernoff--Stein lemma gives, for every $0<\varepsilon<1$,
\begin{equation}\label{eq:geodesic-Stein-exponent-updated}
 \lim_{n\to\infty}-\frac1n\log\beta_{n,\varepsilon}^{(t)}
 =D^t(\rho\Vert\sigma).
\end{equation}
If $V_t(\rho\Vert\sigma)>0$, Strassen's second-order expansion gives
\begin{equation}\label{eq:second-order-Stein-updated}
 -\log\beta_{n,\varepsilon}^{(t)}
 =nD^t(\rho\Vert\sigma)
 +\sqrt{nV_t(\rho\Vert\sigma)}\,\mathcal{N}^{-1}(\varepsilon)
 +O(\log n),
\end{equation}
where $\mathcal{N}$ is the standard normal distribution function \cite{strassen1962}.

\item If we denote the Chernoff information of the canonical experiment for each $t\in [0,1]$ as
\begin{equation}\label{eq:Chernoff-information-classical}
 C_t(\rho,\sigma)
 :=-\min_{0\leq s\leq1}\Psi_{\rho,\sigma}(t,s),
\end{equation}
then Chernoff's theorem gives
\begin{equation}\label{eq:restricted-Chernoff-exponent}
 \lim_{n\to\infty}-\frac1n\log P_{\mathrm e,n}^{(t)}
 =C_t(\rho,\sigma).
\end{equation}

\item The Legendre--Fenchel transform
\begin{equation}\label{eq:rate-function-q-updated}
 I_{t,q}(a)
 :=\sup_{s\in\mathbb R}
 \{sa-\Psi_{\rho,\sigma}(t,s)\}
\end{equation}
is the Cram\'er rate function of the empirical log-likelihood under $q_t$.  Under $p_t$ the rate function is
\begin{equation}\label{eq:rate-function-p-updated}
 I_{t,p}(a)
=I_{t,q}(a)-a.
\end{equation}
Furthermore,
\begin{align}
 I_{t,q}(\partial_s\Psi_{\rho,\sigma}(t,s))
 &=s\,\partial_s\Psi_{\rho,\sigma}(t,s)-\Psi_{\rho,\sigma}(t,s),\notag\\
 I_{t,p}(\partial_s\Psi_{\rho,\sigma}(t,s))
 &=(s-1)\partial_s\Psi_{\rho,\sigma}(t,s)-\Psi_{\rho,\sigma}(t,s).
 \label{eq:supporting-line-tradeoff}
\end{align}
If $p_t\neq q_t$, the function $s\mapsto\Psi_{\rho,\sigma}(t,s)$ is strictly convex, its minimum is attained at the unique point $s_t^*\in(0,1)$ satisfying $\partial_s\Psi_{\rho,\sigma}(t,s_t^*)=0$, and
\begin{equation*}
 I_{t,q}(0)=I_{t,p}(0)=C_t(\rho,\sigma).
\end{equation*}
\end{enumerate}
\end{proposition}

\begin{proof}
Fix $t\in[0,1]$ and abbreviate
\begin{equation*}
 \Psi_t(s):=\Psi_{\rho,\sigma}(t,s),
 \qquad
 \ell_t(z):=\log\frac{p_t(z)}{q_t(z)}.
\end{equation*}
All sums below are taken over the finite set $\mathcal Z_t^+$.  In particular,
$p_t(z),q_t(z)>0$ on this set, so $\ell_t$ is finite and has exponential moments of every
order.  Recall that by \eqref{eq:Psi_log_expectation_exponential},
\begin{equation}\label{eq:proof-Psi-log-partition}
 \Psi_t(s)
 =\log\sum_{z\in\mathcal Z_t^+}q_t(z)e^{s\ell_t(z)}
 =\log\bE_{q_t}[e^{s\ell_t}],
\end{equation}
while \eqref{eq:canonical-tilted-distribution} is given by
\begin{equation}\label{eq:proof-tilted-family}
 \pi_{t,s}(z)
 =\exp\!\bigl(s\ell_t(z)-\Psi_t(s)\bigr)q_t(z).
\end{equation}
Thus $s\mapsto\pi_{t,s}$ is the one-dimensional exponential family with sufficient
statistic $\ell_t$ and log-partition function $\Psi_t$.

The standard differentiation formulas for a log-partition function give its first
derivative as the mean of the sufficient statistic and its second derivative as its
covariance; see
\cite[Proposition 3.1]{wainwrightJordan2008}.  Applying those formulas to
\eqref{eq:proof-Psi-log-partition}--\eqref{eq:proof-tilted-family} yields
\begin{equation*}
 \Psi_t'(s)=\bE_{\pi_{t,s}}[\ell_t],
 \qquad
 \Psi_t''(s)=\Var_{\pi_{t,s}}(\ell_t),
\end{equation*}
which proves \eqref{eq:Psi-first-second-derivative-classical}.  Moreover,
\eqref{eq:Psi-normalization-boundary} and \eqref{eq:proof-tilted-family} imply
$\pi_{t,1}=p_t$.  Therefore
\begin{align}\label{eq:Derivatives_psi_1}
 \Psi_t'(1)
 &=\bE_{p_t}\!\left[\log\frac{p_t}{q_t}\right]
   =D_{\mathrm{KL}}(p_t\Vert q_t),\\
 \Psi_t''(1)
 &=\Var_{p_t}\!\left(\log\frac{p_t}{q_t}\right).
\end{align}
Finally, \eqref{eq:exact-classicalization-updated} with the generator
$f(x)=x\log x$ identifies
$D_{\mathrm{KL}}(p_t\Vert q_t)=D^t(\rho\Vert\sigma)$.  This proves
\eqref{eq:Dt-Vt-from-Psi} and \eqref{eq:Vt-definition-updated}.

We next prove the testing statements.  Proposition \ref{prop:exact-classical-test-reduction}
shows, at every blocklength $n$, that the restricted tests in $\cA_t^{\otimes n}$ are
in one-to-one correspondence with randomized tests for
$p_t^{\otimes n}$ against $q_t^{\otimes n}$ and that both error probabilities are
preserved.  Consequently,
\begin{equation}\label{eq:proof-beta-classical-identification}
 \beta_{n,\varepsilon}^{(t)}
 =\beta_\varepsilon\!\left(p_t^{\otimes n}\Vert q_t^{\otimes n}\right),
\end{equation}
where the right-hand side is the ordinary classical optimal type-II error under the
constraint that the type-I error is at most $\varepsilon$.  The classical
Chernoff--Stein lemma
\cite[Theorem 11.8.3]{coverThomas2006} applied to
\eqref{eq:proof-beta-classical-identification} gives
\begin{equation*}
 \lim_{n\to\infty}-\frac1n\log\beta_{n,\varepsilon}^{(t)}
 =D_{\mathrm{KL}}(p_t\Vert q_t)
 =D^t(\rho\Vert\sigma),
\end{equation*}
which is \eqref{eq:geodesic-Stein-exponent-updated}.

Assume now that $V_t(\rho\Vert\sigma)>0$.  Represent the two classical distributions as
commuting density matrices
\begin{equation*}
 \widehat p_t:=\sum_zp_t(z)|z\rangle\!\langle z|,
 \qquad
 \widehat q_t:=\sum_zq_t(z)|z\rangle\!\langle z|.
\end{equation*}
For this commuting pair,
\begin{equation*}
 D(\widehat p_t\Vert\widehat q_t)
 =D_{\mathrm{KL}}(p_t\Vert q_t)=D^t(\rho\Vert\sigma),
 \qquad
 V(\widehat p_t\Vert\widehat q_t)=V_t(\rho\Vert\sigma).
\end{equation*}
The  bounds in
\cite[Theorem 5]{li2014secondorder} applied to this commuting
pair, give a lower bound with an $O(1)$ remainder and an upper bound with a
$2\log n+O(1)$ remainder.  Together with
\eqref{eq:proof-beta-classical-identification}, they imply
\begin{equation*}
 -\log\beta_{n,\varepsilon}^{(t)}
 =nD^t(\rho\Vert\sigma)
 +\sqrt{nV_t(\rho\Vert\sigma)}\,\mathcal{N}^{-1}(\varepsilon)
 +O(\log n),
\end{equation*}
which proves \eqref{eq:second-order-Stein-updated}.  This is the  finite commuting case of the modern quantum second-order expansions; see also
\cite[Eq. (34)]{tomamichelHayashi2013}.  The original classical
result is due to   Strassen
\cite{strassen1962}.

The same exact reduction applies to the equal-prior Bayesian error
\begin{equation*}
 P_{\mathrm e,n}^{(t)}
 =P_{\mathrm e}\!\left(p_t^{\otimes n},q_t^{\otimes n}\right).
\end{equation*}
Chernoff's theorem
(see e.g. \cite{chernoff1952} or \cite[Theorem 11.9.1]{coverThomas2006}) therefore gives
\begin{align*}
 \lim_{n\to\infty}-\frac1n\log P_{\mathrm e,n}^{(t)}
 &=-\min_{0\leq s\leq1}
   \log\sum_zp_t(z)^sq_t(z)^{1-s}\\
 &=-\min_{0\leq s\leq1}\Psi_t(s)
 =C_t(\rho,\sigma),
\end{align*}
where the second equality is \eqref{eq:Psi-classical-cumulant-updated}.  This proves
\eqref{eq:restricted-Chernoff-exponent}.

It remains to prove the large-deviation statements.  Let
$Z_1,Z_2,\ldots$ be independent random variables with common distribution $q_t$, and set
\begin{equation*}
 L_{t,n}:=\frac1n\sum_{k=1}^n\ell_t(Z_k).
\end{equation*}
By \eqref{eq:proof-Psi-log-partition}, the logarithmic moment-generating function of
$\ell_t(Z_1)$ is exactly $\Psi_t$.  Since $\ell_t$ has finite support, this function is
finite on all of $\mathbb R$.  Cram\'er's theorem
\cite[Theorem 2.2.3]{demboZeitouni1998} therefore shows that
$L_{t,n}$ satisfies a large-deviation principle under $q_t$ with rate function
\begin{equation*}
 I_{t,q}(a)=\sup_{s\in\mathbb R}\{sa-\Psi_t(s)\}\, .
\end{equation*}

Under $p_t$, the logarithmic moment-generating function is obtained from the likelihood
change of measure:
\begin{align*}
 \log\bE_{p_t}[e^{s\ell_t}]
 &=\log\sum_zp_t(z)
   \left(\frac{p_t(z)}{q_t(z)}\right)^s\\
 &=\log\sum_zp_t(z)^{s+1}q_t(z)^{-s}
 =\Psi_t(s+1).
\end{align*}
A second application of Cram\'er's theorem thus gives
\begin{equation*}
 I_{t,p}(a)=\sup_{s\in\mathbb R}\{sa-\Psi_t(s+1)\}.
\end{equation*}
With the change of variables $u=s+1$,
\begin{align*}
 I_{t,p}(a)
 &=\sup_{u\in\mathbb R}\{(u-1)a-\Psi_t(u)\}\\
 &=I_{t,q}(a)-a,
\end{align*}
which is \eqref{eq:rate-function-p-updated}.

The equality case in the Legendre duality for a differentiable
cumulant-generating function
\cite[Lemma 2.2.5(c)]{demboZeitouni1998} shows that 
\begin{equation*}
 I_{t,q}(\Psi_t'(s))=s\,\Psi_t'(s)-\Psi_t(s).
\end{equation*}
Combining this identity with $I_{t,p}(a)=I_{t,q}(a)-a$ gives
\begin{equation*}
 I_{t,p}(\Psi_t'(s))=(s-1)\Psi_t'(s)-\Psi_t(s),
\end{equation*}
proving \eqref{eq:supporting-line-tradeoff}.

Finally, suppose that $p_t\neq q_t$.  Then $\ell_t$ cannot be constant: if
$\ell_t(z)=c$ on $\mathcal Z_t^+$, then $p_t=e^cq_t$, and normalization forces
$e^c=1$ and hence $p_t=q_t$.  Since every tilted distribution $\pi_{t,s}$ has full support on
$\mathcal Z_t^+$, the variance identity already proved gives
\begin{equation*}
 \Psi_t''(s)=\Var_{\pi_{t,s}}(\ell_t)>0,
\end{equation*}
so $\Psi_t$ is strictly convex. At the endpoints, using \eqref{eq:Derivatives_psi_1} and the fact that $\pi_{t,0}=q_t$, we obtain
\begin{equation*}
 \Psi_t'(0)=-D_{\mathrm{KL}}(q_t\Vert p_t)<0,
 \qquad
 \Psi_t'(1)=D_{\mathrm{KL}}(p_t\Vert q_t)>0\, .
\end{equation*}
  Since $\Psi_t'$ is
continuous and strictly increasing, there is a unique $s_t^*\in(0,1)$ satisfying
$\Psi_t'(s_t^*)=0$, and strict convexity makes it the unique global minimizer of
$\Psi_t$.  Consequently,
\begin{align*}
 I_{t,q}(0)
 &=\sup_{s\in\mathbb R}\{-\Psi_t(s)\}
  =-\Psi_t(s_t^*)
  =C_t(\rho,\sigma),\\
 I_{t,p}(0)
 &=I_{t,q}(0)-0
  =C_t(\rho,\sigma).
\end{align*}
This proves the final assertion.
\end{proof}

\begin{remark}
   At $t=1$, the BS relative entropy is the Stein exponent of the canonical classical experiment $(p_1,q_1)$.
\end{remark}

The content of Proposition \ref{prop:classical-likelihood-calculus} is classical once $(p_t,q_t)$ has been constructed.  Its role here is to show that one horizontal section $s\mapsto\Psi_{\rho,\sigma}(t,s)$ organizes all the standard quantities of the selected experiment: its endpoint slope is the KL divergence, its endpoint curvature is the information variance, its minimum is the Chernoff information and its Legendre transform gives the likelihood large-deviation tradeoff. To sum up, given $\rho,\sigma \in \mathcal{S}(\mathcal{H})$, we can construct for every $t \in [0,1]$ classical probabilities that provide the information of the experiment, i.e.
\begin{equation*}
 (\rho,\sigma)
 \xrightarrow{\ t\ }
 (p_t,q_t)
 \xrightarrow{\ s\ }
 \pi_{t,s}
\end{equation*}
and the exact identification of the resulting classical quantities with the geodesic quantum divergences.

\section{Horospherical geometry of the likelihood surface}
\label{sec:horospherical-likelihood-geometry}

Theorem  \ref{theo:geodesic-likelihood-surface} and Section \ref{subsec:classical-statistical-consequences} show that, for every fixed $t$, the function $s\mapsto\Psi_{\rho,\sigma}(t,s)$ is the log-likelihood cumulant of the canonical experiment $(p_t,q_t)$. The purpose of this section is to complement this statistical interpretation with the geometry contained in
\begin{equation*}
 \Psi_{\rho,\sigma}(t,s)
 =\log\langle\Omega_\sigma,\Gamma_t^s\Omega_\sigma\rangle_{\HS}.
\end{equation*}
Notice for the moment we have not properly used that both $t\mapsto \Gamma_t$ and $s \mapsto \Gamma_t^s=I\#_s\Gamma_t$ are both geodesics on the Riemannian manifold of positive operators. These two geodesic structures  contain a large relation of geometric properties when looking at the right manifold.

 We are going to construct a geometric  function $b_\sigma$ on the cone of positive operators such that
\begin{equation*}
 b_\sigma(\Gamma_t^s)=\Psi_{\rho,\sigma}(t,s).
\end{equation*}
This function will be a Busemann function.  We first recall the geometric picture in Euclidean space and then transfer it to the positive cone.

\subsection{The  Euclidean Busemann function}

Before studying the geometry of the likelihood surface in the positive cone, it is going to be useful to visualize the necessary concepts in the Euclidean model.  Given $v\in\mathbb R^d$ a unit vector, consider the geodesic ray $\gamma_v(u)=uv$\footnote{Recall that a geodesic ray on a metric space $(X,d)$ is a path $\gamma:[0,\infty)\to X$ such that $d(\gamma(t),\gamma(t'))=\vert t-t'\vert$ for every $t,t' \in [0,\infty)$.}.  For a fixed $x\in\mathbb R^d$,
\begin{equation*}
 \|x-uv\|-u=-\langle x,v\rangle+o(1),
 \qquad u\longrightarrow\infty \, ,
\end{equation*}
and the limit
\begin{equation*}
 b_v(x):=\lim_{u\to\infty}\bigl(\|x-\gamma_v(u)\|-u\bigr)
 =-\langle x,v\rangle
\end{equation*}
is the Euclidean Busemann function associated with the direction $v$; see \cite[Example 8.24]{bridsonHaefliger1999}. Its level set at height $c$ is
\begin{equation*}
 \{x:b_v(x)=c\}
 =\{x:\langle x,v\rangle=-c\},
\end{equation*}
which is the hyperplane perpendicular to $v$. Moreover,
\begin{equation*}
 \nabla b_v=-v,
 \qquad \|\nabla b_v\|=1.
\end{equation*}
Hence the gradient gives the unit normal to all these parallel hyperplanes.  We can then think of a  Busemann function as a height function attached to a chosen direction at infinity. 

The same construction works on a Hadamard manifold, that is, a complete and simply connected Riemannian manifold with nonpositive sectional curvature. By the Cartan--Hadamard theorem, any two points are joined by a unique geodesic; see \cite[Theorem 4.1]{bridsonHaefliger1999}. If $\gamma:[0,\infty)\to M$ is a unit-speed geodesic ray, its Busemann function is defined as

\begin{equation*}
 b_\gamma(x)
 :=\lim_{u\to\infty}\bigl(d(x,\gamma(u))-u\bigr)\, ,
\end{equation*}
and the existence of the limit follows by the  triangle inequality and can be found in   \cite[Lemma 8.18(1)]{bridsonHaefliger1999}. Moreover, the function $b_\gamma$ is also convex and $1$-Lipschitz; see \cite[Proposition 8.22]{bridsonHaefliger1999}. Its level sets
\begin{equation*}
 \{x:b_\gamma(x)=c\}
\end{equation*}
are called horospheres, and the sublevel sets $\{b_\gamma\le c\}$ are the corresponding horoballs; see \cite[Definitions 8.14 and 8.17]{bridsonHaefliger1999}.

For Hadamard manifolds, Busemann functions $b_\gamma$ are $C^2$ \cite{heintze1977geometry}, convex and its gradient has unit length \cite[Section 2]{itohKimParkSatoh2017}. Since a
horosphere is a level set of $b_{\gamma}$, its gradient is perpendicular to that level set and therefore gives
a natural unit normal. We orient the normal by declaring that $\nabla b_{\gamma}$ points toward increasing
values of $b_{\gamma}$. Equivalently, for the horoball $\{b_{\gamma} \leq  c\}$, the vector $\nabla b_{\gamma}$ points outward and $-\nabla b_{\gamma}$ points inward.

\subsection{The horospheres in the positive cone}

 The cone of  positive  operators $\mathbb{P}(\mathcal{B}_2(\mathcal{H}))$ is a Hadamard manifold with the metric  given by \eqref{eq:metric}. Simply connectedness follows from its convexity, the nonpositive sectional curvature was proven  in \cite{bhatia2003exponential} and the geodesically completeness in \cite[Appendix A.4]{schwartzman2016lognormal}. We can  therefore use Busemann functions on this cone to identify our log-likelihood cumulant function $\Psi(t,s)$ as a  Busemann function of the purification $\Omega_{\sigma}$. These log-expectations functions have already been identified before  in \cite[Section 6]{hiraiNieuwboerWalter2026} as Busemann functions when studying the Kempf-Ness function or in \cite{solan2025geometric} with more generality in Hadamard manifolds. We add here a brief proof for the sake of completeness.

\begin{proposition}
    Let $\sigma \in \mathcal{S}(\mathcal{H})$ be an invertible quantum state and consider the unit vector  $\Omega_\sigma=\sigma^{1/2}$ in $\mathcal B_2(\cH)$. Denote $ P_\sigma=|\Omega_\sigma\rangle\!\langle\Omega_\sigma|$ that corresponds to the orthogonal projection onto the line $\mathbb C\Omega_\sigma$. Then
\begin{equation*}
 \eta_\sigma(u):=e^{-uP_\sigma},
 \qquad u\ge0
\end{equation*}
is a geodesic ray and \begin{equation}\label{eq:Busemann-normalization-profile-new}
 b_\sigma(A)
 =\lim_{u\to\infty}
 \bigl(d_2(A,\eta_\sigma(u))-u\bigr)
 =\log\langle\Omega_\sigma,A\Omega_\sigma\rangle_{\HS}\, ,
\end{equation}
for every operator $A>0$.
\end{proposition}
\begin{proof}
  The operator $\eta_{\sigma}(u)$ acts as   multiplication by $e^{-u}$ on $\Omega_\sigma$ and as the identity on $\Omega_\sigma^\perp$. To prove that it is a geodesic ray compute
  \begin{equation}
      \log(\eta_{\sigma}(u)^{-1/2}\eta_{\sigma}(v)\eta_{\sigma(u)}^{-1/2})=(u-v)P_{\sigma}
  \end{equation}
so
\begin{equation}
    d_2(\eta_{\sigma}(u),\eta_{\sigma}(v))=\vert u-v\vert \Vert P_{\sigma}\Vert_{\HS}=\vert u-v\vert\, ,
\end{equation}
which shows the first statement.

To show \eqref{eq:Busemann-normalization-profile-new}, set $M_u:=A^{-1/2}\eta_\sigma(u)A^{-1/2}$ and let $a=\langle \Omega_{\sigma},A \Omega_{\sigma}\rangle_{\HS}>0$. The distance $d_2(A,\eta_\sigma(u))$ is the Hilbert--Schmidt norm of $\log M_u$. The key observation is that, as $u\to\infty$, exactly one eigenvalue of $M_u$ tends to zero. To determine its size, consider the inverse:
\begin{align*}
 M_u^{-1}
 &=A^{1/2}\eta_\sigma(u)^{-1}A^{1/2}\\
 &=A+(e^u-1)|A^{1/2}\Omega_{\sigma}\rangle\!\langle A^{1/2} \Omega_{\sigma}|\, ,
\end{align*}
where in the  last equation we use that 
\begin{equation*}
    \eta_{\sigma}(u)^{-1}=e^{uP_{\sigma}}=I+\left( \sum_{n=1}^{\infty} \frac{u^n}{n!}\right)P_{\sigma}=I+(e^u-1)P_{\sigma}\, .
\end{equation*}

If $w=A^{1/2}\Omega_{\sigma}$, then $\|w\|^2=a$. Evaluating the Rayleigh quotient on $w/\|w\|$,  we can upper and lower bound the maximum eigenvalue of $M_u^{-1}$ as follows:
\begin{equation*}
 (e^u-1)a+\lambda_{\min}(A)
 \le \lambda_{\max}(M_u^{-1})
 \le (e^u-1)a+\lambda_{\max}(A).
\end{equation*}
Therefore
\begin{equation*}
 \lambda_{\max}(M_u^{-1})=e^u a\,(1+o(1)),
\end{equation*}
and the smallest eigenvalue of $M_u$ satisfies
\begin{equation*}
 \mu_1(u)=\frac{e^{-u}}{a}\,(1+o(1)).
\end{equation*}
The remaining eigenvalues stay bounded away from both zero and infinity. Indeed, as $u\to \infty$, $
 M_u\longrightarrow A^{-1/2}(I-P_\sigma)A^{-1/2}$, 
and this limiting operator has one zero eigenvalue and all its other eigenvalues strictly positive, since $A^{-1/2}$ is invertible. Consequently,
\begin{equation*}
 \log\mu_1(u)=-u-\log a+o(1),
\end{equation*}
whereas the logarithms of the remaining eigenvalues are $O(1)$. It follows that
\begin{equation*}
 d_2(A,\eta_\sigma(u))
 =u+\log a+o(1).
\end{equation*}
Subtracting $u$ proves \eqref{eq:Busemann-normalization-profile-new}.
\end{proof}

Let
\begin{equation*}
 \mathcal H_c^\sigma:=\{A>0:b_\sigma(A)=c\}
\end{equation*}
for the horosphere of height $c$, and
\begin{equation*}
 \mathcal B_c^\sigma:=\{A>0:b_\sigma(A)\le c\}
\end{equation*}
for the corresponding horoball. We can now locate the relative-operator geodesic. Recall that 
\begin{equation*}
 \Gamma_t=L_{\rho^{1-t}}R_{K_t},
 \qquad
 K_t=\sigma^{-1/2}\rho^t\sigma^{-1/2}.
\end{equation*}
Applying $\Gamma_t$ to $\Omega_\sigma=\sigma^{1/2}$ gives $\rho \sigma^{-1/2}$ as we saw in Theorem \ref{theo:geodesic-likelihood-surface}. Therefore,
\begin{equation}
\langle\Omega_\sigma,\Gamma_t\Omega_\sigma\rangle_{\HS}
 =\tr\!\left(\sigma^{1/2}\rho\sigma^{-1/2}\right)
 =\tr\rho=1\, .
\end{equation}
Substituting this identity into the Busemann formula yields

\begin{equation}\label{eq:likelihood-geodesic-horosphere-new}
 b_\sigma(\Gamma_t)
 =\log\langle\Omega_\sigma,\Gamma_t\Omega_\sigma\rangle_{\HS}
 =0,
 \qquad 0\le t\le1.
\end{equation}

As a consequence, the geodesic $\Gamma_t=\Delta_{\rho\vert \sigma}\#_t \widehat{\Delta}_{\rho\vert \sigma}$ lies on the single horosphere $\mathcal H_0^\sigma$. Differentiating $b_\sigma(\Gamma_t)=0$ gives
\begin{equation*}
 0=\frac{\dd}{\dd t}b_\sigma(\Gamma_t)
 =g_{\Gamma_t}\!\left(\nabla b_\sigma(\Gamma_t),
                        \dot\Gamma_t\right).
\end{equation*}
Since the gradient is normal to the horosphere, we obtain that  $\dot\Gamma_t$ is tangent to $\mathcal H_0^\sigma$. 

Differentiating once more,
\begin{equation*}
 \frac{\dd^2}{\dd t^2}b_{\sigma}(\Gamma_t)
 =\operatorname{Hess}b_\sigma|_{\Gamma_t}(\dot \Gamma_t,\dot \Gamma_t)
 +g_{\Gamma_t}\!\left(\nabla b_{\sigma},\nabla_{\dot \Gamma_t}\dot \Gamma_t\right).
\end{equation*}
Here the left-hand side vanishes because $b_\sigma(\Gamma_t)$ is the zero function, and the second term vanishes because a geodesic satisfies $\nabla_{\dot\Gamma_t}\dot\Gamma_t=0$. Then,
\begin{equation}\label{eq:t-null-busemann-hessian-new}
 \operatorname{Hess}b_\sigma|_{\Gamma_t}
 (\dot\Gamma_t,\dot\Gamma_t)=0.
\end{equation}
 In other words, the horosphere has zero normal curvature in the $t$-direction, i.e.  as $\Gamma_t$ moves along the horosphere.

\subsection{A geometric meaning for the statistical quantities of the likelihood surface}

We  will move the coordinate $s$ from $I$ to $\Gamma_t$ through the radial geodesic $s\mapsto\Gamma_t^s$ and will see how this curve moves across the horospherical levels.  For fixed $t$, set
\begin{equation*}
 \gamma_t(s):=I\#_s\Gamma_t=\Gamma_t^s,
 \qquad 0\leq s\leq1.
\end{equation*}
Both endpoints belong to $\mathcal H_0^\sigma$, because $b_\sigma(I)=b_\sigma(\Gamma_t)=0$.  Along the  geodesic, however,
\begin{equation}\label{eq:Psi-as-Busemann-profile-new}
 b_\sigma(\gamma_t(s))
 =\log\langle\Omega_\sigma,\Gamma_t^s\Omega_\sigma\rangle_{\mathrm{HS}}
 =\Psi_{\rho,\sigma}(t,s).
\end{equation}
This  allow us to provide geometric meanings to the quantities studied in  Proposition \ref{prop:classical-likelihood-calculus}.

 Indeed, the chain rule and \eqref{eq:Psi-first-second-derivative-classical} give
\begin{equation}\label{eq:normal-velocity-mean-loglikelihood-new}
 \partial_s\Psi_{\rho,\sigma}(t,s)
 =g_{\gamma_t(s)}\left(
 \nabla b_\sigma(\gamma_t(s)),
 \dot\gamma_t(s)
 \right)
 =\bE_{\pi_{t,s}}[\ell_t].
\end{equation}
as was pointed out in \cite[Section 6]{hiraiNieuwboerWalter2026} for general observables. Since $\nabla b_{\sigma}$ is the normal vector to the horosphere,  $\bE_{\pi_{t,s}}[\ell_t]$ represents the velocity of $\gamma_t(s)$ through the horospheres.  At $s=0$ this velocity is
\begin{equation*}
 \partial_s\Psi_{\rho,\sigma}(t,0)
 =-D_{\mathrm{KL}}(q_t\Vert p_t)\leq0 \, .
\end{equation*}
In particular, for sufficiently small $s$,  $b_{\sigma}(\gamma_t(s))=b_{\sigma}(\gamma_t(0))+s\partial_s\Psi_{\rho,\sigma}(t,0)+o(s)$, and since $b_{\sigma}(\gamma_t(0))=0$, then  $b_{\sigma}(\gamma_t(s))<0$, so the radial geodesic initially enters the  horoball $\mathcal B_0^\sigma$, while at $s=1$ it returns to the boundary horosphere with outward velocity
\begin{equation*}
 \partial_s\Psi_{\rho,\sigma}(t,1)
 =D^t(\rho\Vert\sigma)\geq0.
\end{equation*}

For fixed $t$, $\gamma_t(s)=\Gamma_t^s=e^{s \log \Gamma_t}$ with derivative $\dot \gamma_t(s)=(\log \Gamma_t)\Gamma_t^s$ and has velocity
\begin{align}
     L_t&:=\sqrt{g_{\gamma_t(s)}(\dot \gamma_t(s),\dot \gamma_t(s))}\\ & =\sqrt{\tr[\Gamma_t^{-s}(\log \Gamma_t)\Gamma_t^s\Gamma_t^{-s}(\log \Gamma_t)\Gamma_t^s]}\\ &=\Vert \log \Gamma_t \Vert_{\HS}\\
     &=d_2(I,\Gamma_t)
\end{align}

Whenever $\Gamma_t\neq I$, if $\theta_t(s)$ denotes the angle between $\dot\gamma_t(s)$ and the Busemann normal $\nabla b_\sigma$, then \eqref{eq:normal-velocity-mean-loglikelihood-new} becomes
\begin{equation*}
 \partial_s\Psi_{\rho,\sigma}(t,s)=L_t\cos\theta_t(s).
\end{equation*}
In particular,
\begin{equation}\label{eq:Stein-normal-component-new}
 D^t(\rho\Vert\sigma)
 =L_t\cos\theta_t(1)\, .
\end{equation}
The geodesic divergence is therefore the normal component of the velocity when the curve returns to the  horosphere.

Because $\gamma_t$ is a geodesic, differentiating \eqref{eq:Psi-as-Busemann-profile-new} twice gives
\begin{equation}\label{eq:variance-Busemann-Hessian-new}
 \partial_s^2\Psi_{\rho,\sigma}(t,s)
 =\operatorname{Hess}b_\sigma|_{\gamma_t(s)}
 (\dot\gamma_t(s),\dot\gamma_t(s))
 =\Var_{\pi_{t,s}}(\ell_t).
\end{equation}
Thus the variance  measures how  $b_\sigma(\gamma_t(s))$ bends along the radial geodesic $s\mapsto\gamma_t(s)$. Furthermore,  the convexity given by $\partial_s^2\Psi\geq0$, together with the zero endpoint values for $s=0$ and $s=1$, also shows geometrically that the whole radial geodesic lies in the horoball  $\mathcal B_0^\sigma$. This connection between the Hessian of the  Busemann function and the variance, was also previously identified in \cite[Section 6]{hiraiNieuwboerWalter2026} for general observables, and here we make it specific for our geodesic $\gamma_t(s)$, identifying that this operator variance is exactly the likelihood variance of the canonical experiment constructed in Theorem \ref{theo:geodesic-likelihood-surface}.

Assume now that $p_t\neq q_t$.  Then the likelihood cumulant is strictly convex and has a unique Chernoff minimizer $s_t^*\in(0,1)$.  Since by Proposition \ref{prop:classical-likelihood-calculus} (iv)
\begin{equation*}
 \partial_s\Psi_{\rho,\sigma}(t,s_t^*)=0,
\end{equation*}
equation \eqref{eq:normal-velocity-mean-loglikelihood-new} shows that $\dot\gamma_t(s_t^*)$ is tangent to the horosphere through $\Gamma_t^{s_t^*}$.  Moreover, by Proposition \ref{prop:classical-likelihood-calculus} (iv)
\begin{equation*}
 b_\sigma(\Gamma_t^{s_t^*})
 =\Psi_{\rho,\sigma}(t,s_t^*)
 =-C_t(\rho,\sigma).
\end{equation*}
The Chernoff point $\Gamma_t^{s_t^*}$  is therefore the unique point at which the radial geodesic is tangent to the horosphere $\mathcal H_{-C_t(\rho,\sigma)}^\sigma$.  This intuition motivates the next result, which provides, as a consequence, a new  geometric meaning for the Chernoff information, see Figure \ref{fig:horosphere}.

\begin{theorem}[Horospherical depth]\label{theo:Chernoff-Distance}
    For invertible states $\rho,\sigma \in \mathcal{S}(\mathcal{H}
    )$ and every $t,s \in [0,1]$
    \begin{equation}
        d_2(\Gamma_t^{s},\mathcal{H}_0^{\sigma})=-\Psi_{\rho,\sigma}(t,s)\, .
    \end{equation}
     As a consequence, if the associated probabilities $p_t$ and $q_t$ in Theorem \ref{theo:geodesic-likelihood-surface} satisfy $p_t\neq q_t$,  the Chernoff point $\Gamma_t^{s_t^*}$  is the point on the geodesic  $\gamma_{t}(s)$, $s \in [0,1]$, most distant from the horosphere $\mathcal{H}_0^{\sigma}$. i.e.
   \begin{equation}
    C_t(\rho,\sigma)=\max_{s \in [0,1]}d_2(\Gamma_t^s, \mathcal{H}_0^{\sigma})\, .
\end{equation}
\end{theorem}
\begin{proof}
    As we already mentioned, a Busemann function is $1$-Lipschitz, meaning that for every $A,B>0$, $\vert b_{\sigma}(A)-b_{\sigma}(B)\vert\leq d_2(A,B)$. Let $A=\Gamma_t^{s}$ and $B\in \cH_0^{\sigma}$, so $b_{\sigma}(A)=\Psi_{\rho,\sigma}(t,s)$ and $b_{\sigma}(B)=0$. Since $s\mapsto \Psi_{\rho,\sigma}(t,s)$ is convex for $s \in [0,1]$ and $\Psi_{\rho,\sigma}(t,0)=\Psi_{\rho,\sigma}(t,1)=0$, $\Psi_{\rho,\sigma}(t,s)\leq 0$.  Therefore,
    \begin{equation}
        -\Psi_{\rho,\sigma}(t,s)=\vert b_{\sigma}(\Gamma_t^{s})-b_{\sigma}(B)\vert \leq d_2(\Gamma_t^{s},B )\, .
    \end{equation}
Since this holds for every $B$, taking the infimum we obtain that $-\Psi_{\rho,\sigma}(t,s)\leq d_2(\Gamma_t^{s},\mathcal{H}_0^{\sigma} )$\, .

    To show the other direction, let $\delta(r)$ be a path starting at $\Gamma_t^{s}$ and following the Busemann gradient, i.e.
    \begin{equation}
        \dot \delta(r)=\nabla b_{\sigma}(\delta(r))\, , \quad \delta(0)=\Gamma_t^{s}\, .
    \end{equation}
    By the chain rule,
    \begin{align}
        \frac{\dd}{\dd r}b_{\sigma}(\delta(r))& =g_{\delta(r)}\left( \nabla b_{\sigma}(\delta(r)), \dot \delta (r) \right)\\
        &= \Vert \nabla b_{\sigma}(\delta(r))\Vert^2\\
        &=1\, ,
    \end{align}
    and integrating,
    \begin{equation}
        b_{\sigma}(\delta(r))=b_{\sigma}(\delta(0))+r
    \end{equation}
In particular, if we choose $r=-\Psi_{\rho,\sigma}(t,s)\equiv -\Psi$, then $b_{\sigma}(\delta(-\Psi))=0$, so $\delta(-\Psi) \in \mathcal{H}_0^{\sigma}$.

Since $\Vert \dot \delta(r)\Vert=1$, $\delta(r)$ is parameterized by the arc-length, so there exists a path from $\Gamma_t^{s}$ to $\mathcal{H}_0^{\sigma}$ of length $-\Psi$, which implies that $d_2(\Gamma_t^{s},\cH_0^{\sigma})\leq -\Psi$.

For the second part, since
\begin{equation}
    d_2(\Gamma_t^{s_t^*},\mathcal{H}_0^{\sigma})=-\Psi_{\rho,\sigma}(t,s_t^*)=C_t(\rho,\sigma)\, ,
\end{equation}
whenever $p_t\neq q_t$, the strict convexity of $s\mapsto \Psi(t,s)$ makes the maximizer unique. Since $s_t^*$ is the unique minimizer of $s \mapsto \Psi_{\rho,\sigma}(t,s)$ by Proposition \ref{prop:classical-likelihood-calculus} (iv), we conclude that 
\begin{equation}
    C_t(\rho,\sigma)=-\min_{s\in [0,1]}\Psi_{\rho,\sigma}(t,s)=\max_{s \in [0,1]}d_2(\Gamma_t^s, \mathcal{H}_0^{\sigma})\, .
\end{equation}

\end{proof}

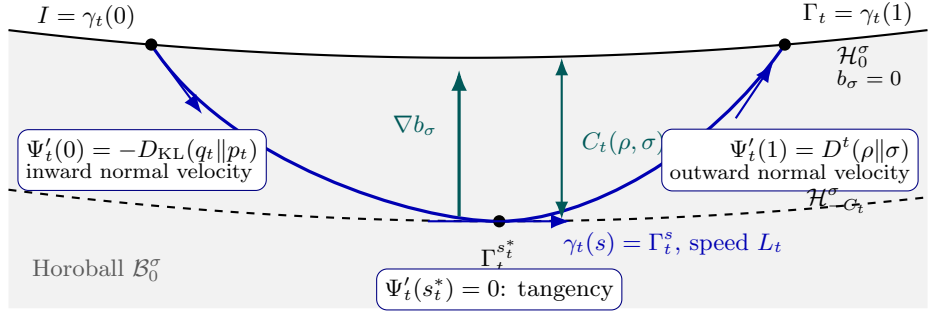
\begin{figure}[t]
\centering
\color{black}
\begin{tikzpicture}[x=1.18cm,y=0.92cm,>=Latex,font=\footnotesize]
  \path[fill=gray!10]
    plot[smooth,domain=-5.15:5.15,samples=100]
      (\x,{3.15+0.015*\x*\x})
    -- (5.15,-0.45) -- (-5.15,-0.45) -- cycle;
  \draw[thick]
    plot[smooth,domain=-5.15:5.15,samples=100]
      (\x,{3.15+0.015*\x*\x});
  \draw[thick,dashed]
    plot[smooth,domain=-5.15:5.15,samples=100]
      (\x,{0.80+0.015*(\x-0.35)*(\x-0.35)});

  \draw[very thick,blue!70!black]
    (-3.55,3.34)
    .. controls (-2.70,1.65) and (-0.60,0.80) .. (0.35,0.80)
    .. controls (1.30,0.80) and (2.70,1.75) .. (3.55,3.34);

  \fill (-3.55,3.34) circle (2.3pt);
  \fill (3.55,3.34) circle (2.3pt);
  \fill (0.35,0.80) circle (2.3pt);
  \node[above left=2pt] at (-3.55,3.34) {$I=\gamma_t(0)$};
  \node[above right=3pt] at (3.55,3.34) {$\Gamma_t=\gamma_t(1)$};
  \node[below=3pt] at (0.35,0.80) {$\Gamma_t^{s_t^*}$};

  \draw[-{Latex[length=2.6mm]},blue!70!black,thick]
    (-3.42,3.10) -- (-2.98,2.34);
  \draw[-{Latex[length=2.6mm]},blue!70!black,thick]
    (3.00,2.28) -- (3.42,3.08);
  \draw[-{Latex[length=2.6mm]},blue!70!black,thick]
    (-0.45,0.80) -- (1.15,0.80);

  \draw[-{Latex[length=2.8mm]},teal!70!black,very thick]
    (-0.10,0.87) -- (-0.10,2.98);
  \node[teal!70!black,anchor=east]
    at (-0.23,2.18) {$\nabla b_\sigma$};
  \draw[<->,teal!70!black,thick]
    (1.05,0.83) -- (1.05,3.15);
  \node[teal!70!black,anchor=west]
    at (1.16,1.92) {$C_t(\rho,\sigma)$};

  \node[draw=blue!50!black,rounded corners,fill=white,align=left,
        inner sep=3pt,anchor=west]
    at (-5.05,1.68)
    {$\Psi_t'(0)=-D_{\rm KL}(q_t\Vert p_t)$\\[-1mm]
     {\scriptsize inward normal velocity}};
  \node[draw=blue!50!black,rounded corners,fill=white,align=right,
        inner sep=3pt,anchor=east]
    at (5.05,1.68)
    {$\Psi_t'(1)=D^t(\rho\Vert\sigma)$\\[-1mm]
     {\scriptsize outward normal velocity}};
  \node[draw=blue!50!black,rounded corners,fill=white,
        inner sep=3pt]
    at (0.35,-0.18)
    {$\Psi_t'(s_t^*)=0$: tangency};

  \node[anchor=west,align=left] at (4.02,3.02)
    {$\mathcal H_0^\sigma$\\[-1mm]{\scriptsize $b_\sigma=0$}};
  \node[anchor=west] at (3.68,1.10)
    {$\mathcal H_{-C_t}^\sigma$};
  \node[gray!70!black,anchor=west]
    at (-5.00,0.10)
    {Horoball $\mathcal B_0^\sigma$};
  \node[blue!70!black] at (2.30,0.48)
    {$\gamma_t(s)=\Gamma_t^s$, speed $L_t$};
\end{tikzpicture}
\caption{Example of the geometry of the horospheres: The radial geodesic $\gamma_t(s)$ has constant total speed $L_t$, enters the horoball with normal component $-D_{\mathrm{KL}}(q_t\Vert p_t)$, is tangent to the deepest horosphere at the value of parameter $s_t^*$, and returns with normal component $D^t(\rho\Vert\sigma)$.  The normal distance from the Chernoff point to $\mathcal H_0^\sigma$ is $C_t(\rho,\sigma)$.  }
\label{fig:horosphere}
\end{figure}

\section*{Acknowledgments}
 Á.C. acknowledges support from the Deutsche Forschungsgemeinschaft
(DFG, German Research Foundation), Project-ID 470903074, TRR 352.
This project was funded within the QuantERA II Programme, which
received funding from the European Union's Horizon 2020 research and
innovation programme under Grant Agreement No. 101017733. PCR acknowledges support from the European
Research Council (ERC Grant Agreement No.\ 948139) and from the Excellence Cluster --
Matter and Light for Quantum Computing (ML4Q-2).

\paragraph{AI Statement:}
The authors acknowledge the use of  ChatGPT 5.6 Sol to assist with brainstorming, drafting and organizing the literature review. The final output was thoroughly verified by the authors, who take full responsibility for the accuracy and integrity of this work.

\bibliographystyle{abbrv}
\bibliography{literatura}

\appendix

\section{Proofs of the lemmas}\label{app:proofs-of-lemmas}

For ease of reference, we collect in this appendix the complete statements and proofs of all lemmas used in the main text. 

\begin{lemma}\label{lemma:proof-orthogonality-relations}
Let $\rho=\sum_i p_iE_{ii}>0$, and  let  $D$ be a  traceless diagonal self-adjoint matrix in the eigenbasis of $\rho$. Then, for $i<j$ and $k<l$,
\begin{align*}
     g_{\rho,t}^{(f)}(D,A_{ij})
&=g_{\rho,t}^{(f)}(D,B_{ij})=0,\\
         g_{\rho,t}^{(f)}(A_{ij},A_{kl})
   &=g_{\rho,t}^{(f)}(B_{ij},B_{kl})
    =g_{\rho,t}^{(f)}(A_{ij},B_{kl})=0,
    \qquad (i,j)\neq(k,l),\\
g_{\rho,t}^{(f)}(A_{ij},B_{ij})&=0,\\
         g_{\rho,t}^{(f)}(A_{ij},A_{ij})
   &=g_{\rho,t}^{(f)}(B_{ij},B_{ij}).
\end{align*}
\end{lemma}
\begin{proof}
As already mentioned, the tangent space of the set of quantum states is formed by  traceless self-adjoint matrices. Such a matrix $X$ can  be decomposed as
\begin{equation*}
    X=D_X+\sum_{i<j}(a_{ij}A_{ij}+b_{ij}B_{ij})\, ,
\end{equation*}
where $D_X$ is traceless and  diagonal, $A_{ij}=E_{ij}+E_{ji}$ and $B_{ij}=\mathrm{i}(E_{ij}-E_{ji})$.

Since the divergence is unitarily invariant \cite[Proposition 4.3]{capel2026geodesic}, differentiating twice on the diagonal gives the covariance relation
\begin{equation*}
   g_{U\rho U^*,t}^{(f)}(UXU^*,UYU^*)
    =g_{\rho,t}^{(f)}(X,Y).
  \end{equation*}

Take the  eigenbasis of $\rho$. In this basis $\rho$ is diagonal and $E_{ij}$ has a single nonzero entry, equal to $1$, in position $(i,j)$. Consider  a  diagonal unitary matrix
\begin{equation*}
 U_{\varphi}=\sum_k e^{\mathrm{i}\varphi_k}E_{kk}.
\end{equation*}
The  conjugation of $E_{ij}$ by $U_{\varphi}$ gives
\begin{equation*}
    U_{\varphi}E_{ij}U_{\varphi}^*
    =e^{\mathrm{i}(\varphi_i-\varphi_j)}E_{ij}.
\end{equation*}
Choosing $\varphi_i-\varphi_j=\pi$, we obtain
\begin{equation*}
 U_{\varphi}A_{ij}U_{\varphi}^*=-A_{ij},
 \qquad
 U_{\varphi}B_{ij}U_{\varphi}^*=-B_{ij},
\end{equation*}
whereas $U_{\varphi}DU_{\varphi}^*=D$ for every diagonal $D$. By unitary covariance and real bilinearity,
\begin{equation*}
    g_{\rho,t}^{(f)}(D,A_{ij})
    =g_{\rho,t}^{(f)}(D,-A_{ij})
    =-g_{\rho,t}^{(f)}(D,A_{ij})\, .
\end{equation*}
Therefore $g_{\rho,t}^{(f)}(D,A_{ij})=0$. The same argument gives $g_{\rho,t}^{(f)}(D,B_{ij})=0$, proving \eqref{eq:Orthogonality_D}.

For ordered pairs $(i,j)$ and $(k,l)$, unitary covariance and sesquilinearity give
\begin{equation}\label{eq:Phases_metric}
    g_{\rho,t}^{(f)}(E_{ij},E_{kl})
    =e^{-\mathrm{i}(\varphi_i-\varphi_j)}
     e^{\mathrm{i}(\varphi_k-\varphi_l)}
     g_{\rho,t}^{(f)}(E_{ij},E_{kl})\,  .
\end{equation}
For $g_{\rho,t}^{(f)}(E_{ij},E_{kl})$ to be nonzero, the phase factor in \eqref{eq:Phases_metric} must equal one for every choice of the phases. This occurs exactly when $(i,j)=(k,l)$. Consequently, if $i<j$, $k<l$,  and $(i,j)\neq(k,l)$, every matrix-unit pairing arising from the expansions of $A_{ij},B_{ij}$ and $A_{kl},B_{kl}$ vanishes. This proves \eqref{eq:orthogonality_pairs}.

Finally, choose $\varphi_i-\varphi_j=\pi/2$. Then
\begin{equation*}
 U_{\varphi}A_{ij}U_{\varphi}^*=B_{ij},
 \qquad
 U_{\varphi}B_{ij}U_{\varphi}^*=-A_{ij}.
\end{equation*}
Using unitary covariance and symmetry on the real tangent space,
\begin{align*}
 g_{\rho,t}^{(f)}(A_{ij},B_{ij})
 &=g_{\rho,t}^{(f)}(B_{ij},-A_{ij})
 =-g_{\rho,t}^{(f)}(A_{ij},B_{ij}),
\end{align*}
which proves \eqref{eq:Orthogonality_A_B}. The same transformation gives
\begin{equation*}
    g_{\rho,t}^{(f)}(A_{ij},A_{ij})
    =g_{\rho,t}^{(f)}(B_{ij},B_{ij})\, ,
\end{equation*}
which is \eqref{eq:Same_norm_AB}.
\end{proof}

\begin{lemma}\label{lemma:proof-two-level-reduction}

     Let $\rho\in \cS(\cH)$ be an invertible quantum state with eigenvalues  $(p_i)_i$ and corresponding eigenvectors $(v_i)_i$. Let $E_{ij}=\vert v_i \rangle \langle v_j \vert$ and define the traceless self-adjoint matrix $A_{ij}=E_{ij}+E_{ji}$. Then, for any operator convex function $f:(0,\infty)\to \mathbb{R}$ with $f(1)=0$,
\begin{equation*}
    g_{\rho,t}^{(f)}(A_{ij},A_{ij})=\frac{1}{p_j}g_{P_r,t}^{(f)}(A,A)\, ,
\end{equation*}
where, for $r=p_i/p_j$,
\begin{equation*}
         P_r=\begin{pmatrix}
            r & 0\\
             0 & 1
         \end{pmatrix},
         \qquad
         A=\begin{pmatrix}
            0 & 1\\
            1 & 0
         \end{pmatrix}.
\end{equation*}
\end{lemma}
\begin{proof}
     Let $v_i,v_j$ denote two  different  eigenvectors of $\rho$ and let  $\cH_{i,j}=\operatorname{span}\{v_i,v_j\}$, so $\cH=\cH_{i,j}\oplus \cH_{i,j}^{\perp}$. Write $\rho=P\oplus \tau$, where
    \begin{equation*}
         P=\begin{pmatrix}
             p_i & 0\\
              0 & p_j
         \end{pmatrix}\,  , \qquad
        \tau=\sum_{k\neq i,j}p_k \vert v_k \rangle \langle v_k\vert .
    \end{equation*}
    Consider the rotation on the $i,j$ subspace $U_{\theta}=R_{\theta}\oplus I_{\cH_{i,j}^{\perp}}$, where
     \begin{equation*}
         R_{\theta}=\begin{pmatrix}
               \cos \theta &  -\sin\theta\\
             \sin \theta & \cos \theta
        \end{pmatrix}\, .
     \end{equation*}

Define $\rho_{\theta}:=U_{\theta}\rho  U_{\theta}^*=P_{\theta}\oplus \tau$, where $P_\theta=R_\theta P R_\theta^*$.  We first show that
\begin{equation}\label{eq:Direct_sum_divergence}
 D_f^t(\rho\Vert \rho_{\theta})=D_f^t(P\Vert P_{\theta}).
\end{equation}
Since $\rho$ is invertible, $\rho^a=P^a\oplus \tau^a$ for every $a \in \mathbb{R}$, and  similarly for $\rho_{\theta}$. Hence
\begin{equation*}
    \rho_{\theta}^{-1/2}\rho^t \rho_{\theta}^{-1/2}
    =(P_{\theta}^{-1/2}P^tP_{\theta}^{-1/2})
      \oplus(\tau^{-1/2}\tau^t \tau^{-1/2})
    =Z_{\theta}\oplus \tau^{t-1}\, ,
\end{equation*}
where $Z_{\theta}=P_{\theta}^{-1/2}P^tP_{\theta}^{-1/2}$. Therefore, the interpolating  relative operator  can be written as
\begin{equation}\label{eq:geodesic_Lemma_direct_sum}
    \gamma_{\rho,\rho_{\theta}}(t)
    =L_{P^{1-t}\oplus \tau^{1-t}}
     R_{Z_{\theta}\oplus \tau^{t-1}}\, .
\end{equation}

The decomposition $\cH=\cH_{i,j}\oplus \cH_{i,j}^{\perp}$ allows us to write any $X \in \cB(\cH)$ as
\begin{equation*}
    X=\begin{pmatrix}
        X_{11} & X_{12} \\
        X_{21} & X_{22}
     \end{pmatrix}.
\end{equation*}
Because  the left and right multipliers in \eqref{eq:geodesic_Lemma_direct_sum} are block diagonal, $\gamma_{\rho,\rho_{\theta}}(t)$ preserves each  of the four block subspaces. Since this relative operator is self-adjoint on the Hilbert--Schmidt space, these subspaces are reducing  and are consequently preserved by $f(\gamma_{\rho,\rho_{\theta}}(t))$. As
\begin{equation*}
 \rho_\theta^{1/2}=P_\theta^{1/2}\oplus\tau^{1/2},
\end{equation*}
the divergence splits into two contributions:
\begin{align*}
    D_f^t(\rho\Vert \rho_{\theta})
    &=\langle \rho_{\theta}^{1/2},f(\gamma_{\rho,\rho_{\theta}}(t)) \rho_{\theta}^{1/2}\rangle_{\mathrm{HS}}\\
     &=D_f^t(P\Vert P_{\theta})
      +\langle \tau^{1/2}, f(L_{\tau^{1-t}}R_{\tau^{t-1}})\tau^{1/2}\rangle_{\mathrm{HS}}\\
    &=D_f^t(P\Vert P_{\theta})\, ,
\end{align*}
since $\tau^{1/2}$ is an eigenvector of $L_{\tau^{1-t}}R_{\tau^{t-1}}$  with eigenvalue $1$ and $f(1)=0$. This proves \eqref{eq:Direct_sum_divergence}.

To compute the Hessian, notice that by the standard second-variation identity for divergence functions \cite[pp. 632--633]{eguchi1992geometry}, \cite[Proposition 2.7]{matsuzoe2025divergence}, or \cite[Eq. (4)]{ay2015novel}, the metric induced by a smooth divergence can be evaluated by varying only the second entry:
\begin{equation}\label{eq:Second_derivative_divergence}
    g_{\eta,t}^{(f)}(X,X)
    =\left. \frac{\dd^2}{\dd \theta^2}
    D_f^t(\eta\Vert \sigma_{\theta})\right|_{\theta=0},
\end{equation}
for every smooth  path satisfying $\sigma_0=\eta$ and $\dot\sigma_0=X$.

Direct multiplication gives
\begin{equation*}
    P_{\theta}=\begin{pmatrix}
         p_i \cos^2\theta+p_j \sin^2\theta
         & (p_i-p_j)\sin\theta \cos\theta\\
        (p_i-p_j)\sin\theta \cos\theta
        & p_i  \sin^2\theta+p_j \cos^2\theta
    \end{pmatrix},
\end{equation*}
and differentiation at $\theta=0$ yields
\begin{equation*}
    \dot P_0
     =\begin{pmatrix}
        0 & p_i-p_j\\
        p_i-p_j  & 0
    \end{pmatrix}
    =(p_i-p_j)A.
\end{equation*}
Likewise,
\begin{equation*}
 \dot\rho_0=(p_i-p_j)A_{ij}.
\end{equation*}
Applying \eqref{eq:Second_derivative_divergence} to the paths $\rho_\theta$ and $P_\theta$, and using \eqref{eq:Direct_sum_divergence}, we obtain
\begin{align}\label{eq:Before_cancel}
 (p_i-p_j)^2g_{\rho,t}^{(f)}(A_{ij},A_{ij})
 &=\left.\frac{\dd^2}{\dd\theta^2}D_f^t(\rho\Vert\rho_\theta)\right|_{\theta=0}\\
 &=\left.\frac{\dd^2}{\dd\theta^2}D_f^t(P\Vert P_\theta)\right|_{\theta=0}\notag\\
 &=(p_i-p_j)^2g_{P,t}^{(f)}(A,A).\notag
\end{align}
If $p_i\neq p_j$, cancellation in \eqref{eq:Before_cancel} gives
\begin{equation}\label{eq:Reduction_metric}
  g_{\rho,t}^{(f)}(A_{ij},A_{ij})=g_{P,t}^{(f)}(A,A).
\end{equation}

  When $p_i=p_j$, the  rotation path  is constant and cannot be used to cancel the  prefactor. In that case, for sufficiently small $s,u$,
\begin{equation*}
   \rho+sA_{ij}=(P+sA)\oplus\tau,
      \qquad
     \rho+uA_{ij}=(P+uA)\oplus\tau.
\end{equation*}
The same reducing-subspace  argument used  above gives
\begin{equation*}
    D_f^t(\rho+sA_{ij}\Vert\rho+uA_{ij})
     =D_f^t(P+sA\Vert P+uA),
\end{equation*}
    and taking the mixed derivative at $s=u=0$ again yields \eqref{eq:Reduction_metric}. Thus \eqref{eq:Reduction_metric} holds without any nondegeneracy assumption.

It remains to discuss  the scaling of the divergence. For $\lambda>0$,
\begin{align*}
     \gamma_{\lambda\rho,\lambda\sigma}(t)
    &=L_{(\lambda\rho)^{1-t}}
R_{(\lambda\sigma)^{-1/2}(\lambda\rho)^t(\lambda\sigma)^{-1/2}}\\
        &=L_{\lambda^{1-t}\rho^{1-t}}
    R_{\lambda^{t-1}\sigma^{-1/2}\rho^t\sigma^{-1/2}}\\
    &=\gamma_{\rho,\sigma}(t).
 \end{align*}
Since $(\lambda\sigma)^{1/2}=\lambda^{1/2}\sigma^{1/2}$, it follows that
\begin{equation*}
   D_f^t(\lambda \rho\Vert \lambda \sigma)
    =\lambda D_f^t(\rho\Vert\sigma).
 \end{equation*}
    Consequently,
\begin{equation*}
    D_f^t(\lambda \rho +sX\Vert \lambda \rho + uY)
    =\lambda D_f^t\left(
    \rho+\frac{s}{\lambda}X
\left\Vert
         \rho+\frac{u}{\lambda}Y
   \right.\right)\, .
\end{equation*}
  Taking derivatives in $s$ and $u$ at zero, each differentiation supplies a factor $\lambda^{-1}$,  and therefore
\begin{equation}\label{eq:Scaling_Metric_in_rho}
         g_{\lambda \rho,t}^{(f)}(X,Y)=\frac{1}{\lambda}g_{\rho,t}^{(f)}(X,Y)\, .
\end{equation}
  Writing
\begin{equation*}
P=p_j\begin{pmatrix}
        p_i/p_j & 0\\
   0 & 1
    \end{pmatrix}=p_jP_r
 \end{equation*}
     and applying \eqref{eq:Scaling_Metric_in_rho} to \eqref{eq:Reduction_metric} proves \eqref{eq:Equation_Statement_Lemma_Metric_Basis}.

\end{proof}

\begin{lemma}\label{lemma:Taylor}
Let $0\leq t<1$,  $r>0$, $r\neq1$, and
\begin{equation*}
   P=\begin{pmatrix}r&0\\0&1\end{pmatrix},
    \qquad
    P_\theta=R_\theta P R_\theta^*,
\qquad
        Z_\theta=P_\theta^{-1/2}P^tP_\theta^{-1/2}.
\end{equation*}
Let $a_\theta$ and $b_\theta$ denote the first and second rows of $P_\theta^{1/2}$, respectively. If $f(1)=f'(1)=0$, then
\begin{equation*}\label{eq:Taylor_result}
D_f^t(P\Vert  P_\theta)
          =(r-1)^2
   \frac{f(r^{1-t})+r^{1-2t}f(r^{t-1})}
    {(r^{1-t}-1)^2}\,\theta^2
     +o(\theta^2).
\end{equation*}
\end{lemma}
\begin{proof}
Put
\begin{equation*}
 h:=r^{1/2},\qquad q:=r^{1-t}.
\end{equation*}
Since $P_{\theta}=R_{\theta}PR_{\theta}^*$, the spectral theorem shows that $P_{\theta}^{1/2}=R_{\theta}P^{1/2}R_{\theta}^*$, which yields
\begin{equation*}
    P_{\theta}^{1/2}=\begin{pmatrix}
        h\cos^2 \theta+\sin^2\theta
        & (h-1)\sin \theta \cos \theta \\
        (h-1)\sin \theta \cos \theta
         & h\sin^2 \theta+\cos^2\theta
     \end{pmatrix}
    =\begin{pmatrix}
        h & (h-1)\theta\\
        (h-1)\theta & 1
    \end{pmatrix}+O(\theta^2 )\, .
\end{equation*}
The vectors $a_{\theta}$ and $b_{\theta}$ can therefore be expanded as
\begin{equation}\label{eq:Expansion_a_b}
    a_{\theta}=\begin{pmatrix}h\\(h-1)\theta\end{pmatrix}+O(\theta^2),
    \qquad
    b_{\theta}=\begin{pmatrix}(h-1)\theta\\1\end{pmatrix}+O(\theta^2).
\end{equation}

To expand $Z_\theta$, one must use the inverse square root rather than the square root. Functional calculus gives
\begin{equation*}
     P_\theta^{-1/2}
    =R_\theta\begin{pmatrix}h^{-1}&0\\0&1\end{pmatrix}R_\theta^*
=B_0+\theta  B_1+O(\theta^2),
\end{equation*}
  where
\begin{equation*}
    B_0=\begin{pmatrix}h^{-1}&0\\0&1\end{pmatrix},
\qquad
        B_1=(h^{-1}-1)\begin{pmatrix}0&1\\1&0\end{pmatrix}.
\end{equation*}
Therefore
\begin{align*}
Z_\theta
        &=(B_0+\theta B_1)P^t(B_0+\theta B_1)+O(\theta^2)\\
   &=B_0P^tB_0
    +\theta(B_1P^tB_0+B_0P^tB_1)+O(\theta^2).
 \end{align*}
Since $P^t=\operatorname{diag}(r^t,1)$ and $r^t=h^2/q$, direct multiplication yields
\begin{equation}\label{eq:Expansion_Z}
    Z_\theta
    =\begin{pmatrix}q^{-1}&v\theta\\v\theta&1\end{pmatrix}
    +O(\theta^2),
\qquad
          v=-\frac{(h-1)(q+h)}{hq}.
\end{equation}

    The row-sector decomposition used in the theorem is purely algebraic and  applies to the present centered function $f$; hence
\begin{equation}\label{eq:Rows_decomposition}
    D_f^t(P\Vert P_\theta)
=\langle a_\theta,f(qZ_\theta)a_\theta\rangle
         +\langle b_\theta,f(Z_\theta)b_\theta\rangle.
\end{equation}
Let
\begin{equation*}
A_{\theta}:=qZ_{\theta}.
\end{equation*}
Then
\begin{equation}\label{eq:Expansion_A_theta}
    A_{\theta}
=\begin{pmatrix}
        1 & qv \theta\\
   qv \theta & q
    \end{pmatrix}+ O(\theta^2).
 \end{equation}
At $\theta=0$, $A_0$ is diagonal with the  distinct eigenvalues $1$ and $q$, because $r>0$, $r\neq1$, and $t<1$. Thus, for sufficiently small $\theta$, the two eigenvalues remain  simple and  may be chosen smoothly. Since
\begin{equation*}
   A_0'=\begin{pmatrix}0&qv\\qv&0\end{pmatrix},
\end{equation*}
the standard first-order perturbation formula for a simple eigenvalue \cite{kato1995} gives
\begin{equation*}
         \lambda_1'(0)=\langle e_1,A_0'e_1\rangle=0,
   \qquad
    \lambda_q'(0)=\langle e_2,A_0'e_2\rangle=0.
 \end{equation*}
    Consequently,
\begin{equation}\label{eq:Eigenvalues_A_theta}
   \lambda_1(\theta)=1+O(\theta^2),
    \qquad
     \lambda_q(\theta)=q+O(\theta^2).
\end{equation}
Because $f(1)=f'(1)=0$, Taylor's theorem at $1$ gives
\begin{equation}\label{eq:Order_lambda_one}
    f(\lambda_1(\theta))=O(\theta^4).
 \end{equation}
Thus the eigenvalue branch near $1$ does not contribute to the coefficient of  $\theta^2$.

    For the normalized eigenvector corresponding to $\lambda_q(\theta)$, consider the expansion
\begin{equation}\label{eq:Eigenvector_wq}
    w_q(\theta)=e_2+\alpha \theta e_1+O(\theta^2)\,  .
\end{equation}
    The absence of a first-order component along $e_2$ follows from normalization. On the one hand, \eqref{eq:Expansion_A_theta} gives
\begin{equation*}
    A_{\theta}w_q(\theta)
    =q e_2+\theta(qv+\alpha)e_1+O(\theta^2)\, ,
\end{equation*}
  and, on the other hand, \eqref{eq:Eigenvalues_A_theta} gives
\begin{equation*}
    \lambda_q(\theta)w_q(\theta)
     =q e_2+q\alpha \theta e_1+O(\theta^2)\, .
\end{equation*}
Comparing the coefficients of $e_1$ yields
\begin{equation}\label{eq:alpha_value}
     \alpha=\frac{qv}{q-1}\, .
 \end{equation}
Let $w_1(\theta)$ be the normalized eigenvector corresponding to $\lambda_1(\theta)$. By the spectral theorem,
\begin{align*}
   \langle a_\theta,f(A_\theta)a_\theta\rangle
    &=f(\lambda_1(\theta))
    |\langle w_1(\theta),a_\theta\rangle|^2
+f(\lambda_q(\theta))
        |\langle w_q(\theta),a_\theta\rangle|^2.
\end{align*}
The first  term  is $O(\theta^4)$ by \eqref{eq:Order_lambda_one} and boundedness of $a_\theta$. Moreover,
\begin{equation*}
f(\lambda_q(\theta))=f(q)+O(\theta^2),
\end{equation*}
and, by \eqref{eq:Expansion_a_b}, \eqref{eq:Eigenvector_wq}, and \eqref{eq:alpha_value},
\begin{align*}
    \langle w_q(\theta),a_\theta\rangle
&=\left(h-1+\frac{hqv}{q-1}\right)\theta+O(\theta^2).
\end{align*}
Therefore
\begin{equation}\label{eq:expression_atheta}
     \langle  a_\theta,f(A_\theta)a_\theta\rangle
=f(q)\left(h-1+\frac{hqv}{q-1}\right)^2\theta^2
         +o(\theta^2).
\end{equation}

We now treat the second row sector. By \eqref{eq:Expansion_Z}, the eigenvalues of $Z_0$ are  $q^{-1}$  and $1$. The eigenvalue branch near $1$ again contributes only $O(\theta^4)$. Let $\mu_{q^{-1}}(\theta)$ denote the branch near $q^{-1}$ and let  $\psi_{q^{-1}}(\theta)$ be a corresponding normalized eigenvector. Then
\begin{equation*}
        \mu_{q^{-1}}(\theta)=q^{-1}+O(\theta^2),
    \qquad
    \psi_{q^{-1}}(\theta)=e_1+\beta\theta e_2+O(\theta^2).
 \end{equation*}
    Substitution into the  eigenvalue equation
\begin{equation*}
   Z_\theta\psi_{q^{-1}}(\theta)
    =\mu_{q^{-1}}(\theta)\psi_{q^{-1}}(\theta)
 \end{equation*}
  gives
\begin{equation*}
   \beta=\frac{v}{q^{-1}-1}.
\end{equation*}
  Using the spectral theorem exactly as above,  together  with the expansion of $b_\theta$ in \eqref{eq:Expansion_a_b},  gives
\begin{equation}\label{eq:expression_btheta}
        \langle b_{\theta}, f(Z_{\theta})b_{\theta}\rangle
   =f(q^{-1})
    \left(h-1+\frac{v}{q^{-1}-1}\right)^2\theta^2
    +o(\theta^2).
\end{equation}

  Combining \eqref{eq:Rows_decomposition}, \eqref{eq:expression_atheta}, and \eqref{eq:expression_btheta}, we obtain
\begin{align}\label{eq:before_simplifying}
    D_f^t(P\Vert P_\theta)
&=\left[
        f(q)\left(h-1+\frac{hqv}{q-1}\right)^2
   +f(q^{-1})\left(h-1+\frac{v}{q^{-1}-1}\right)^2
    \right]\theta^2
    +o(\theta^2).
\end{align}
Finally, substituting
\begin{equation*}
    v=-\frac{(h-1)(q+h)}{hq}
 \end{equation*}
gives
\begin{align}
   h-1+\frac{hqv}{q-1}
    &=(h-1)-\frac{(h-1)(q+h)}{q-1}
      =-\frac{h^2-1}{q-1}
=-\frac{r-1}{q-1},\label{eq:first_factor}\\
        h-1+\frac{v}{q^{-1}-1}
   &=h-1-\frac{qv}{q-1}
    =\frac{q(h^2-1)}{h(q-1)}
    =\frac{q(r-1)}{h(q-1)}.\label{eq:second_factor}
\end{align}
   Substituting \eqref{eq:first_factor} and \eqref{eq:second_factor} into \eqref{eq:before_simplifying}, and using $h^2=r$ and $q=r^{1-t}$, yields
\begin{align*}
    D_f^t(P\Vert P_\theta)
    &=\left[
\frac{(r-1)^2}{(q-1)^2}f(q)
        +\frac{q^2(r-1)^2}{r(q-1)^2}f(q^{-1})
   \right]\theta^2+o(\theta^2)\\
    &=(r-1)^2
    \frac{f(r^{1-t})+r^{1-2t}f(r^{t-1})}
{(r^{1-t}-1)^2}\,\theta^2
        +o(\theta^2),
\end{align*}
  which concludes the result.

\end{proof}

The scalar proof of the next lemma uses the same type of hyperbolic reparameterization and sign analysis that Besenyei used in the study of the Hasegawa--Petz mean \cite[Proposition 1]{besenyei2012hasegawa}.  Our function $\Phi_t$ is a similar  function but not the same, so we reproduce his proof.

\begin{lemma}\label{lemma:Phi_monotonicity_convexity}
For $u>0$ and $0\leq t<1$, define
\begin{equation*}
       \Phi_t(u):=
       \frac{(1-t)u\sinh u-\sinh((1-t)u)\sinh(tu)}
       {\sinh^2((1-t)u)}.
\end{equation*}
The function admits a continuous extension to $[0,1]\times[0,\infty)$ given at the boundary by
\begin{equation*}
       \Phi_t(0)=1,\qquad \Phi_1(u)=\cosh u.
\end{equation*}
For every fixed $u\geq0$, the map $t\mapsto\Phi_t(u)$ is nondecreasing and convex on $[0,1]$. Moreover, it is strictly increasing whenever $u>0$.
\end{lemma}
\begin{proof}
For $z>0$, define
\begin{equation}\label{eq:Definition_psi}
       \psi(z):=\coth z-\frac{z}{\sinh^2z}.
\end{equation}
Using
\begin{equation*}
 \cosh u\sinh z-\sinh u\cosh z=\sinh(z-u)
\end{equation*}
with $z=(1-t)u$, the defining expression for $\Phi_t(u)$ can be rewritten as
\begin{equation}\label{eq:Phi_psi_formula}
       \Phi_t(u)=\cosh u-\sinh u\,\psi((1-t)u).
\end{equation}
The singularity of $\psi$ at zero is removable. Indeed, the standard expansions of the hyperbolic functions give
\begin{equation*}
       \psi(z)=\frac23z-\frac4{45}z^3+O(z^5).
\end{equation*}
We therefore set $\psi(0)=0$. Formula \eqref{eq:Phi_psi_formula} then extends continuously to $t=1$, where it gives
\begin{equation*}
       \Phi_1(u)=\cosh u.
\end{equation*}
It also extends continuously to $u=0$, where $\Phi_t(0)=1$ for every $t\in[0,1]$.

We next establish the two sign properties of $\psi$ which are needed below. Differentiating \eqref{eq:Definition_psi} gives
\begin{equation*}
        \psi'(z)=\frac{2}{\sinh^2z}\bigl(z\coth z-1\bigr).
\end{equation*}
For $z>0$, the quantity in parentheses is strictly positive. Indeed,
\begin{equation*}
          z\coth z>1
       \quad\Longleftrightarrow\quad
        z\cosh z-\sinh z>0,
\end{equation*}
and the function on the right vanishes at zero and has derivative $z\sinh z>0$. Hence
\begin{equation*}
       \psi'(z)>0\qquad(z>0).
\end{equation*}
A second differentiation yields
\begin{equation*}
       \psi''(z)
        =-\frac{2z(\cosh(2z)+2)-3\sinh(2z)}{\sinh^4z}.
\end{equation*}
The numerator is strictly positive for $z>0$. To see this without using any auxiliary inequality, expand it into its power series:
\begin{equation*}
       2z(\cosh(2z)+2)-3\sinh(2z)
       =\sum_{n=2}^{\infty}
       \frac{2^{2n+2}(n-1)}{(2n+1)!}z^{2n+1}>0.
\end{equation*}
Consequently,
\begin{equation*}
       \psi''(z)<0\qquad(z>0).
\end{equation*}

Fix now $u>0$. From \eqref{eq:Phi_psi_formula}, for $0\leq t<1$,
\begin{equation}\label{eq:Phi_derivative_t}
       \frac{\partial}{\partial t}\Phi_t(u)
        =u\sinh u\,\psi'((1-t)u)>0
\end{equation}
and
\begin{equation*}
       \frac{\partial^2}{\partial t^2}\Phi_t(u)
        =-u^2\sinh u\,\psi''((1-t)u)>0.
\end{equation*}
For $u=0$, the function is identically equal to one. Thus, for every $u\geq0$, the map $t\mapsto\Phi_t(u)$ is nondecreasing and convex on $[0,1)$. Its continuous extension at $t=1$ is therefore nondecreasing and convex on the closed interval $[0,1]$. Moreover, it is strictly increasing whenever $u>0$.
\end{proof}

\end{document}